\documentclass[12pt]{article}

\usepackage{amssymb}
\usepackage{amsmath}
\usepackage{amsfonts}
\usepackage{dsfont}
\usepackage{geometry,enumerate,bbm}
\usepackage{natbib}
\usepackage{epsfig}
\usepackage{float}
\usepackage[onehalfspacing]{setspace}
\usepackage{graphicx,color}
\usepackage{subcaption}
\usepackage{lscape}
\usepackage[flushleft]{threeparttable}
\usepackage{ulem}
\usepackage{scrextend}

\definecolor{dgreen}{rgb}{0,0.5,0}
\RequirePackage[colorlinks,citecolor=blue,urlcolor=blue,linkcolor=blue]{hyperref}
\allowdisplaybreaks
\newtheorem{example}{Example}[section]
\newtheorem{assumption}{Assumption}
\newtheorem{theorem}{Theorem}

\newtheorem{lemma}{Lemma}

\newtheorem{remark}{Remark}[section]
\newenvironment{proof}[1][Proof]{\noindent \textbf{#1.} }{\  \rule{0.5em}{0.5em}}
\begin{document}

\title{\bf Binary choice logit models with general fixed effects for panel and network data\thanks{
This research was
supported by the Gregory C. Chow Econometric Research Program at Princeton
University, by  the Danish National Research Foundation grant number DNRF186 to the Aarhus Center for Econometrics (ACE), and by  the
European Research Council grant ERC-2018-CoG-819086-PANEDA.}}
\author{\setcounter{footnote}{2}
Kevin Dano%
\thanks{Princeton University, \texttt{kdano@princeton.edu} }
\and 
Bo E. Honor{\'e}%
\thanks{Princeton University, \texttt{honore@princeton.edu} }
\and 
 Martin Weidner%
\thanks{
University of Oxford, \texttt{martin.weidner@economics.ox.ac.uk} } }
\date{August 2025}

\maketitle
\thispagestyle{empty}
\setcounter{page}{0}

\bigskip
\begin{abstract}
\noindent
This paper systematically analyzes and reviews identification strategies for binary choice logit models with fixed effects in panel and network data settings. We examine both static and dynamic models with general fixed-effect structures, including individual effects, time trends, and two-way or dyadic effects. A key challenge is the incidental parameter problem, which arises from the increasing number of fixed effects as the sample size grows. We explore two main strategies for eliminating nuisance parameters: conditional likelihood methods, which remove fixed effects by conditioning on sufficient statistics, and moment-based methods, which derive fixed-effect-free moment conditions. We demonstrate how these approaches apply to a variety of models, summarizing key findings from the literature while also presenting new examples and new results.
\end{abstract}

\newpage
 
\section{Introduction}

Binary outcome models with fixed effects have a long history in economics and related fields. They allow for flexible modeling of individual decisions while accounting for unobserved heterogeneity. These models arise naturally in two types of data structures: traditional panel data, where individuals or firms are observed over time, and network or dyadic data, where observations correspond to interactions between pairs of units. In both cases, researchers often wish to estimate the effect of observed covariates on a binary outcome, while allowing for unit-specific or pair-specific unobserved components. Examples include individual labor market choices, firm entry decisions, and the formation of social or economic networks.

A central challenge in these models is the treatment of fixed effects, which capture unobserved heterogeneity. When treated as parameters to be estimated, fixed effects can lead to the incidental parameter problem, first identified by \cite{neyman1948consistent}. This issue arises because the number of fixed effects grows with the sample size.
If the number of observations per parameter increases slowly, as in certain network models or panel data settings with increasing number of time periods, estimators may be asymptotically biased. In such cases, bias correction methods can be applied; see, for example, \cite{HahnNewey2004}, and the surveys by \cite{ArellanoHahn2007} and \cite{FernandezValWeidner2018}.
On the other hand, when the number of parameters grows proportionally with the number of observations, as in standard panel data models with a fixed number of time periods, estimators are often inconsistent. In these settings, the most effective strategy is to eliminate the fixed effects.\footnote{In an alternative approach, \cite{HigginsJochmans2024} show that in parametric settings like the one considered in this paper, it is possible to use the bootstrap to do valid inference on the common parameters using the inconsistent and biased estimator that treats the fixed effects as parameters to be estimated.}

In this paper, we summarize and extend two existing methods for estimating binary choice logit models with fixed effects by eliminating them from the estimation problem. Our focus is on both static and dynamic models, allowing for very general fixed effect structures. We discuss models for standard panel data, as well as dyadic models that arise in the study of networks.

Two main general approaches have been developed to address the incidental parameter problem. The first is based on conditional likelihood and requires a parametric model conditional on the unobservable fixed effects. This method eliminates the fixed effects by conditioning on sufficient statistics. It was first introduced by \cite{rasch1960studies} in the context of educational testing and later adapted to econometric applications by \cite{chamberlain_analysis_1980}. In binary logit models, the fixed effects drop out of the likelihood when we condition on certain linear combinations of the outcomes. This approach has been extended to more complex fixed effect structures, including two-way fixed effects in dyadic data \citep{charbonneau2017multiple} and network formation models \citep{graham2017econometric}. Conditional likelihood methods are attractive because they provide a way to estimate structural parameters without modeling the distribution of the unobserved effects. However, they are only available in models where sufficient statistics can be found in closed form.

The second approach relies on moment conditions that do not depend on the fixed effects.\footnote{In this paper we focus on moment \textsl{equality} conditions. Moment \textsl{inequalities} form the basis for the semiparametric approach in \cite{Manski87}, and have also been used to construct and estimate identified sets for the structural parameters. See, for example, \cite{PakesPorterShepardCalderWang2022}.} This idea is sometimes called functional differencing. It was formalized by \cite{bonhomme2012functional} and applied to discrete choice models by 
\cite{kitazawa2021transformations} and 
\cite{honore2024dynamic}. The basic idea is to construct functions of the data whose conditional expectation is zero for all values of the fixed effects. These functions then define valid moment conditions for the structural parameters. This approach has also been used to develop estimators in dynamic logit models with general fixed effects by \cite{dano2023transition}. Moment-based methods can be used in a wider range of models than conditional likelihood. They are especially useful in dynamic models or in models with more complicated fixed effect structures.

This paper reviews both approaches in the context of binary choice logit models. We present the underlying identification ideas in a unified framework and show how they apply to standard examples. These include static panel models with individual fixed effects, models with heterogeneous time trends, two-way fixed effects, and network formation settings. For dynamic models, we discuss how moment conditions can be constructed under general assumptions, and how these conditions relate to the number of fixed effects in the model. We also give examples from dynamic panel and dynamic network settings. We focus on the logit model, rather than, say, the probit model or a semiparametric version, because it is known from \cite{chamberlain2010binary} that in the simplest case of a static binary response panel data model with two time periods, regular  root-$n$ consistent estimation of the common parameters is only possible in the logit model.\footnote{For a discussion of what can be learned about the common parameters in dynamic binary response models with panel data in a more semiparametric setting, see for example \cite{KHAN2023105515}.}

  The literature has long employed both conditional likelihood and moment condition approaches to estimate common parameters in fixed effects variants of other standard nonlinear models. For instance, \cite{HausmanHallGriliches84} applied conditional likelihood methods to panel data versions of Poisson regression models. Similarly, \cite{Honore92,HONORE199335} and \cite{Hu2002} developed estimators for static and dynamic censored and truncated regression models using moment functions. Additionally, \cite{Chamberlain_1992} and \cite{Wooldridge_1997} utilized moment conditions in a class of static and dynamic multiplicative models. A hybrid approach is taken by \cite{Kyriazidou97, Kyriazidou01}, who combined conditional likelihood and moment conditions to construct estimators for static and dynamic sample selection models.

The fixed effects approach discussed earlier, and further developed in this paper, places no assumptions on the relationship between the individual-specific effects and the explanatory variables. A common alternative is to adopt a random effects framework, which assumes a specific distribution for the individual effects and estimates the model parameters via maximum likelihood. However, this approach faces two key challenges.

First, in many economic applications, explanatory variables are at least partially chosen by individuals. In such cases, it is often implausible to assume independence between these variables and the individual effects (see \citealt{Mundlak1961}). Second, dynamic models introduce the initial conditions problem: the distribution of the initial dependent variable typically depends on both the individual effects and the explanatory variables in a complex way (see \citealt{Heckman81b}).

A solution to these issues, originally proposed by \cite{Mundlak1978}, is the correlated random effects model. This approach specifies a model for the individual effects conditional on the explanatory variables and, in dynamic settings, also on the initial conditions. For further discussion, see \cite{wooldridge2005}.\footnote{For a recent overview of the historical distinction between fixed and random effects in panel data, see \cite{BellemareMillimet2025}.}

 The remainder of the paper is organized as follows.
Section \ref{sec:Static} reviews and extends the conditional likelihood approach for estimating common parameters in static logit models. The key insight from this section is that under relatively simple conditions, one can identify a sufficient statistic for the fixed effects, which then serves as the basis for conditional likelihood estimation of the common parameters—specifically, the coefficients on strictly exogenous explanatory variables.
Section \ref{sec:DynamicsNoX} turns to dynamic models in which the only explanatory variables are lagged outcomes. In this setting, there are again straightforward conditions under which a sufficient statistic for the fixed effects can be found. However, unlike the static case, the resulting conditional likelihood may fail to identify all the common parameters of the model.
This limitation motivates Section \ref{sec:GeneralDynamic}, which explores more general dynamic specifications. Within the logit framework, we demonstrate that it is often possible to construct moment conditions that allow for estimation of the parameters even when no non-trivial sufficient statistic exists, or when conditioning on a sufficient statistic yields a likelihood function that is uninformative about a subset of the parameters of interest. 
Section \ref{sec: Conclusion} concludes.

\section{Static logit models}
\label{sec:Static}

\subsection{Model setup}

We observe data $(Y_t, X_t, w_t)$ for $t = 1, \ldots, T$, where $Y_t \in \{0,1\}$ is a random binary outcome, $X_t \in \mathbb{R}^{d_x}$ is a random vector of covariates whose associated coefficient $\beta \in \mathbb{R}^{d_x}$ is the parameter of interest, and $w_t \in \mathbb{R}^{d_w}$ are non-random vectors associated with an unobserved effect $A \in \mathbb{R}^{d_w}$.\footnote{
Most results extend to random $w_t$, but since all our examples involve non-random $w_t$, we focus on that case.
}
We denote the outcome vector as $Y = (Y_1, \ldots, Y_T)' \in \mathcal{Y} = \{0,1\}^T$, and we write $X = (X_1, \ldots, X_T) \in \mathcal{X} \subset \mathbb{R}^{d_x \times T}$ and $\mathsf{W} = (w_1, \ldots, w_T) \in \mathbb{R}^{d_w \times T}$. We use capital letters to denote random variables (such as $Y_t$, $X_t$, and $A$) and upright font to denote non-random matrices (such as $\mathsf{W}$). The parameter $\beta$ is treated as non-random, while the fixed effect $A$ is allowed to be random.

\begin{assumption} 
\label{ass:StaticModel}
The data-generating process is
\begin{align*}       
{\rm Pr} \left( Y=y \, \big| \, X,A \right)  &=  \prod_{t=1}^T \, 
\left[ \frac {1} {1+\exp(X_t' \, \beta + w_t' \, A) }  
\right]^{1-y_t}
\left[ \frac {\exp(X_t' \, \beta + w_t' \, A)} {1+\exp(X_t' \, \beta + w_t' \, A) }  
\right]^{y_t} .
\end{align*}
\end{assumption}

\medskip
This is a standard assumption for binary logit models with fixed effects. We could write all probability statements explicitly conditional on $\mathsf{W}$ (e.g.\ ${\rm Pr}(Y = y \mid \mathsf{W}, X, A)$ in Assumption~\ref{ass:StaticModel}), but since $\mathsf{W}$ is treated as non-random throughout, we omit this for notational simplicity. 

Since we consider $A \in \mathbb{R}^{d_w}$, Assumption~\ref{ass:StaticModel} implies
that
\begin{align}
    0<{\rm Pr} \left( Y=y \, \big| \, X=x  \right)<1 , \qquad
   \text{ for all }  y \in {\cal Y} \text{ and } x \in {\cal X}.
   \label{BoundedProbability}
\end{align}   
However, all results in this section extend to cases where (components of) $A$ may take values in $\pm \infty$, provided that the boundedness condition in \eqref{BoundedProbability} continues to hold.

\begin{example}[Standard fixed effects in panel data:]
\label{ex:StandardPanel}
    Consider $d_w = 1$ and $w_t = 1$ for all $t \in \{1, \ldots, T\}$. Then
    $$
        X_t' \, \beta + w_t' \, A = X_t' \, \beta + A
    $$
    corresponds to the standard panel logit model with an individual-specific intercept. We observe an i.i.d.\ sample of $(Y,X)$,
    denoted by $\{(Y_i, X_i) : i = 1, \ldots, n\}$, where the data for each unit $i$ satisfy Assumption~\ref{ass:StaticModel}. The index in the logit specification becomes $X_{it}' \, \beta + A_i$.
    $\mathsf{W} = (1, \ldots, 1)$ and $\beta$ are non-random and constant across units $i$.
\end{example}

\begin{example}[Dyadic fixed effects in network data:]
\label{ex:DyadicExample}
    Consider $d_w = n$ and $T = n(n - 1)/2$, where each observation $t = (i,j) = (j,i)$ corresponds to an unordered pair of distinct units $i, j \in \{1, \ldots, n\}$, with $i \neq j$. For each dyad $(i,j)$, we observe a binary outcome $Y_{ij} \in \{0,1\}$ indicating whether a link is present between units $i$ and $j$. Let $X_{ij} \in \mathbb{R}^{d_x}$ denote dyad-level covariates capturing observed characteristics of the pair, such as homophily measures. Define $w_{ij} \in \mathbb{R}^n$ as a selection vector such that\footnote{
    For each unordered pair $t = (i,j)$ with $i \neq j$, define the vector $w_{ij} \in \mathbb{R}^n$ by
$$
        (w_{ij})_k = 
        \begin{cases}
        1, & \text{if } k = i \text{ or } k = j, \\
        0, & \text{otherwise}.
        \end{cases}
$$
    }
$$
        X_{ij}' \, \beta + w_{ij}' \, A = X_{ij}' \, \beta + A_i + A_j
$$
    corresponds to a network formation model with unit-specific fixed effects.
    This specification captures both homophily (through $X_{ij}$) and degree heterogeneity (through $A_i$ and $A_j$), as in \cite{graham2017econometric}. The model can be interpreted as describing a small subnetwork within a larger graph, and by considering many such subnetworks, we obtain multiple independent or weakly dependent observations of $(Y, X)$. For example, 
    \cite{graham2017econometric} considers subnetworks of size $n=4$ to construct his ``tetrad logit'' estimator. Again, $\mathsf{W}$ and $\beta$ are non-random and constant across subnetworks.
\end{example}

These are two classic examples of models satisfying the structure in Assumption~\ref{ass:StaticModel}. Various generalizations are discussed later.

\subsection{Identification via conditional likelihood}

Our goal is to identify the parameter $\beta$ without imposing any assumptions on the distribution of the unobserved effect $A$. In the static binary logit model of Assumption~\ref{ass:StaticModel}, it is well known that  $\mathsf{W} Y$ is a sufficient statistic for $A$, implying that conditioning on this statistic removes dependence on $A$.
Formally, for any two outcome vectors $y_1, y_2 \in \mathcal{Y} = \{0,1\}^T$ such that $\mathsf{W} y_1 = \mathsf{W} y_2$, one can show that
\begin{align}
        \frac{{\rm Pr} \left( Y=y_1 \, \big| \, X,A \right) }
       {{\rm Pr} \left( Y=y_2 \, \big| \, X,A \right) }
       &= \exp\left[\beta' X (y_1-y_2)\right] ,
        \label{SufficientStats1}
\end{align}
which implies that the distribution $Y$
conditional on $Y \in \{y_1, y_2\}$, $X$, and $A$ does not depend on $A$, as long as $\mathsf{W} y_1 = \mathsf{W} y_2$:
\begin{align}
        {\rm Pr} \left( Y=y_1 \, \big| \, X,A, Y \in \{y_1,y_2\} \right)
        &= \frac{\exp\left[\beta' X (y_1-y_2)\right]}
        {1+\exp\left[\beta' X (y_1-y_2)\right]} .
        \label{SufficientStats2}
\end{align}
To construct such identifying pairs $y_1, y_2$, it is useful to reformulate the condition $\mathsf{W} y_1 = \mathsf{W} y_2$ as requiring the existence of a difference vector $w_{\perp} = y_1 - y_2$ satisfying $\mathsf{W} w_{\perp} = 0$. The following lemma formalizes this observation.

\begin{lemma}
   \label{lemma:Wperp}
   There exist
   $y_1,y_2 \in \{0,1\}^T$ with
   $y_1 \neq y_2$ and
    $\mathsf{W} y_1 = \mathsf{W} y_2$
  if and only if  
   there exists $w_{\perp} \in \{-1,0,1\}^T$
   with $w_{\perp} \neq 0$ and  $\mathsf{W} w_{\perp} = 0$.
\end{lemma}

\medskip
The condition $\mathsf{W} w_{\perp} = 0$ is familiar from linear models: if the $T$-vector $Y$ satisfies
$Y = X' \beta + \mathsf{W}' A + \varepsilon$ and $\mathsf{W} w_{\perp} = 0$, then pre-multiplying by $w_{\perp}'$ eliminates the nuisance term $\mathsf{W} A$. Remarkably, the same condition allows us to eliminate $A$ in the binary logit model as well --- provided that $w_{\perp}$ satisfies
$w_{\perp} \in \{-1,0,1\}^T$, which implies that it
corresponds to a difference of two binary vectors.
Ultimately, this leads to the following identification result.

\begin{theorem}
   \label{th:IdentificationStatic}
   Suppose Assumption~\ref{ass:StaticModel} holds. Assume further that for some integer $d_\perp \geq 1$, there exists a non-random matrix $\mathsf{W}_\perp \in \{-1, 0, 1\}^{T \times d_\perp}$ such that $\mathsf{W} \, \mathsf{W}_\perp = 0$, and there is no vector $b \in \mathbb{R}^{d_x}$ for which $b' X \mathsf{W}_\perp = 0$ almost surely.  
   Then, 
   the parameter $\beta$ is uniquely identified
   from the distribution of $(Y,X)$.
\end{theorem}

\medskip
Crucially, this identification result imposes no restrictions on the distribution of $A$ or its dependence with $X$.
The condition $\mathsf{W} \, \mathsf{W}_{\perp} = 0$ ensures that the columns of $\mathsf{W}_{\perp}$ are orthogonal to the fixed effects. The rank condition guarantees that these orthogonal directions yield sufficient variation in $X$ to identify all components of $\beta$. 

Note that the number of columns in $\mathsf{W}_{\perp}$, denoted $d_\perp$, can be smaller than $d_x$. What matters is that the matrix $X \mathsf{W}_{\perp}$ has enough variation to identify $\beta$. In particular, the vector $\beta$ can be identified with $d_\perp = 1$.

\begin{remark}
\label{remark:LinearAnalogy}
The identification condition in Theorem~\ref{th:IdentificationStatic} closely parallels that of a standard linear fixed effects model. Suppose we observe 
$
Y = X' \beta + \mathsf{W}' A + \varepsilon$. In this setting, the parameter $\beta$ is identified if there exists a matrix $\mathsf{W}_{\perp} \in \mathbb{R}^{T \times d_\perp}$ such that $\mathsf{W} \mathsf{W}_{\perp} = 0$ and $X \mathsf{W}_{\perp}$ is non-collinear. Pre-multiplying the linear model equation by $\mathsf{W}_{\perp}'$ yields
$
\mathsf{W}_{\perp}' Y = \mathsf{W}_{\perp}' X' \beta + \mathsf{W}_{\perp}' \varepsilon,
$
in which $A$ no longer appears. The resulting equation can then be estimated by ordinary least squares (OLS).

In the binary choice case, we are subject to the additional restriction that $\mathsf{W}_{\perp}$ must have integer entries in $\{-1,0,1\}$, since it must represent differences between binary outcome vectors. Nonetheless, the identification conditions for the logit model 
in Theorem~\ref{th:IdentificationStatic}
are directly analogous to that of the linear model, modulo this constraint.
\end{remark}

\begin{remark}
An alternative way to understand our identification result is via the population conditional maximum likelihood estimator (CMLE), as in \citet{davezies2024identification}. Given any $s \in \mathbb{R}^{d_w}$, define $\mathcal{S}(s) = \{y \in \{0,1\}^T : \mathsf{W}y = s \}$ and consider the conditional log-likelihood
\[
\ell_c(\beta \mid y,x,s) := \log \left[ {\rm Pr}(Y = y \mid X = x, \mathsf{W}Y = s) \right] ,
\]
which is well-defined and free of $A$ due to sufficiency of $\mathsf{W}Y$. If the expected Hessian
\[
H_{\rm CMLE} := \mathbb{E}\left[ - \frac{\partial^2 \ell_c(\beta \mid Y,X,\mathsf{W}Y)}{\partial \beta \, \partial \beta'} \right]
\]
is positive definite, then $\beta$ is the unique maximizer of the population CMLE objective.
To connect this to Theorem~\ref{th:IdentificationStatic}, observe that
\[
H_{\rm CMLE} =
\mathbb{E}\left[ \sum_{y_1,y_2 \in \mathcal{S}(\mathsf{W}Y)} 
\omega(\mathsf{W}Y, X, y_1, y_2) \,
X (y_1 - y_2) (y_1 - y_2)' X' \right]
\]
for some positive weights $\omega(\mathsf{W}Y, X, y_1, y_2) > 0$. This shows that variation in $X(y_1 - y_2)$ along directions satisfying $\mathsf{W}(y_1 - y_2) = 0$ is essential for identification of $\beta$ under CMLE, just as in our rank condition.
\end{remark}

The presentation here is closely related to classical conditional likelihood approaches (e.g., \citealp{rasch1960studies,andersen1970asymptotic,chamberlain_analysis_1980}), but the formulation in Theorem~\ref{th:IdentificationStatic} highlights a useful structural parallel to linear models. While the sufficiency argument is well known, we are not aware of prior work that frames the construction of identifying pairs through the condition $W w_{\perp} = 0$ with $w_{\perp} \in \{-1,0,1\}^T$. This perspective unifies a number of existing results and provides a clean basis for constructing and analyzing conditional likelihood estimators in settings with general fixed effect structures, as illustrated in the examples that follow.

\subsection{Examples}
We now illustrate Theorem~\ref{th:IdentificationStatic} in several models. Some of the examples are well-known, while others are novel.
In each case, we describe the fixed effect structure and provide explicit constructions of vectors  $w_{\perp}$ (or matrices $\mathsf{W}_{\perp}$) that satisfy the orthogonality condition $W w_{\perp} = 0$. We aim to carry over as much intuition as possible from the linear model analogy in Remark~\ref{remark:LinearAnalogy}.

\subsubsection{Standard fixed effects in panel data}
\label{ex:StaticStandardFE}

This example corresponds to the setup in Example~\ref{ex:StandardPanel}. Recall that $d_w = 1$ and $w_t = 1$ for all $t \in \{1,\ldots,T\}$, so that the index in the logit model is $X_t' \beta + A$, with a scalar fixed effect $A$ shared across all time periods, and $\mathsf{W} = (1, \ldots, 1) \in \mathbb{R}^{1 \times T}$.
To understand the identification strategy, consider first the analogous linear model, making the cross-sectional index $i$ explicit:
$$
Y_{it} = X_{it}' \beta + A_i + \varepsilon_{it}, \qquad t = 1, \ldots, T, \quad i=1,\ldots,n.
$$
Here, we can eliminate $A_i$ by differencing, e.g. for $T = 2$:
$$
Y_{i1} - Y_{i2} = (X_{i1} - X_{i2})' \beta + (\varepsilon_{i1} - \varepsilon_{i2}).
$$
This corresponds to the linear combination $w_{\perp}' Y_i$ with $w_{\perp} = (1, -1)'$.
Crucially, since $w_{\perp} \in \{-1,0,1\}^T$, the same vector works for the logit model:

For the static panel logit with $T=2$ we know that
$
{\rm Pr} \left( Y_i = y \, \big| \, X_i, A_i, Y_{i1} + Y_{i2} = 1 \right)
$
is free of $A_i$. The conditioning event $Y_{i1} + Y_{i2} = 1$ restricts us to the pair $y_1=(1,0)'$ and $y_2=(0,1)'$, for which $w_{\perp} = y_1 - y_2 = (1,-1)'$.
This recovers the classical results of \cite{rasch1960studies} and \cite{andersen1970asymptotic}. For $T > 2$, any vector in $\{-1,0,1\}^T$ with components summing to zero provides a valid $w_{\perp}$. If $X_i w_{\perp}$ varies sufficiently across units, then $\beta$ is identified.

\subsubsection{Heterogeneous time trends}
\label{ex:HeterogeneousTrends}

Let $d_w = 2$, and for each $t \in \{1, \ldots, T\}$, define the regressors associated with the fixed effect as $w_t = (1,t)'$, so that the index becomes
$
X_t' \, \beta + w_t' \, A = X_t' \beta + A_1 + t A_2,
$
where $A_1$ is a unit-specific intercept and $A_2$ a unit-specific time trend. The matrix $\mathsf{W} \in \mathbb{R}^{2 \times T}$ consists of a row of ones and a row of time indices.
Again, we consider the linear model analog,  making the cross-sectional index $i$ explicit:
$$
Y_{it} = X_{it}' \beta + A_{i1} + t A_{i2} + \varepsilon_{it}, \qquad t = 1, \ldots, T, \quad i=1,\ldots,n.
$$
For the linear model, $T=3$ is sufficient to eliminate both fixed effects:
$$
Y_{i1} - 2 Y_{i2} + Y_{i3} = (X_{i1} - 2 X_{i2} + X_{i3})' \beta + (\varepsilon_{i1} - 2\varepsilon_{i2} + \varepsilon_{i3}).
$$
This linear combination corresponds to $w_{\perp}=(1,-2,1)'$, which does not satisfy $w_{\perp} \in \{-1,0,1\}^T$, and is therefore not applicable to the logit model. Indeed, for the static panel logit model with heterogeneous time trends and $T=3$, the coefficient $\beta$ is generally not point-identified.

To eliminate both $A_{i1}$ and $A_{i2}$ in the logit model, we need 
$w_{\perp} \in \{-1,0,1\}^T$ such that $\sum_{t=1}^T w_{\perp,t} = 0$ and $\sum_{t=1}^T w_{\perp,t} \, t = 0$. 
For $T = 4$ this is satisfied for $w_{\perp} = (1, -1, -1, 1)'$,
which implies that in the static logit model with heterogeneous time trends, ${\rm Pr} \left( Y=y_1 \, \big| \, X,A \right) / {\rm Pr} \left( Y=y_2 \, \big| \, X,A \right)$ is independent of $A=(A_1,A_2)$ for $y_1 = (1,0,0,1)'$ and $y_2=(0,1,1,0)'$, since $w_{\perp} = y_1 - y_2$.

More generally, we can consider $w_t=(1,t,t^2,\ldots,t^p)'$, with corresponding index
$$ X_t' \, \beta + w_t' \, A =X_t' \, \beta + A_1 + t \, A_2  + t^2 \, A_3
     + \ldots + t^p A_{p+1} .
   $$
Table~\ref{tab:PanelTimeTrends} shows the minimal number of time periods $T$ and corresponding weight vectors $w_{\perp} \in \{-1,0,1\}^T$ that satisfy $W w_{\perp} = 0$ for this model for $p \leq 5$ (for $p=6$ the minimal $T$ is $31$).
For $p \geq 1$, the minimal weight vectors exhibit an alternating symmetry pattern (even $p$ antisymmetric, odd $p$ symmetric) and can be constructed recursively by placing the $(p-1)$-solution and its reversed copy (negated for antisymmetric cases) with zero-padding and optimal overlap to minimize length while maintaining the constraint $w_{\perp} \in \{-1,0,1\}^T$.

\begin{table}[tb]
\begin{center}
\begin{tabular}{c|c|l}
$p$ & $T$ & $w_{\perp}$ \\
\hline
0 & 2 & (1, -1)' \\
1 & 4 & (1, -1, -1, 1)' \\
2 & 7 & (1, -1, -1, 0, 1, 1, -1)' \\
3 & 12 & (1, -1, -1, 0, 1, 0, 0, 1, 0, -1, -1, 1)' \\
4 & 16 & (1, -1, -1, 0, 0, 1, 1, 1, -1, -1, -1, 0, 0, 1, 1, -1)' \\
5 & 23 & (1, -1, -1, 0, 0, 1, 1, 0, 0, 0, -1, 0, -1, 0, 0, 0, 1, 1, 0, 0, -1, -1, 1)' 
\end{tabular}
\end{center}
\caption{\label{tab:PanelTimeTrends}  The table shows the minimal number of time periods $T$ and the corresponding weight vector $w_{\perp}$ needed to eliminate fixed effects of the form $A_1 + t A_2 + \ldots + t^p A_{p+1}$ in static binary logit panel models.}
\end{table}

\subsubsection{Overlapping fixed effects}
\label{ex:OverlappingEffects}

Consider a setting with overlapping fixed effects where $T=3$ and $d_w=2$. Define
$$
\mathsf{W} = \begin{pmatrix}
1 & 1 & 0 \\
0 & 1 & 1
\end{pmatrix},
$$
which yields the index structure
$$
X_t' \, \beta + w_t' \, A = \begin{cases} 
X_1' \, \beta + A_1 & \text{if } t=1, \\
X_2' \, \beta + A_1 + A_2 & \text{if } t=2, \\
X_3' \, \beta + A_2 & \text{if } t=3.
\end{cases}
$$
Note that observation $t=2$ is affected by both fixed effects, while observations $t=1$ and $t=3$ are each affected by only one. To eliminate both $A_1$ and $A_2$, we can use $w_{\perp}=(1,-1,1)'$, which satisfies $\mathsf{W} w_{\perp} = 0$ since each fixed effect appears with net zero weight. This may be the simplest non-trivial extension of the standard panel fixed effects model to more general fixed effects.

\subsubsection{Two-way fixed effects}
\label{ex:TwoWayPanel}

Consider two-way fixed effects panel models, where each observation is indexed by a unit-time pair $t = (i, \tau)$, with $i \in \{1,\ldots,n\}$ denoting individuals and $\tau \in \{1,\ldots,\mathcal{T}\}$ denoting time periods. The total number of observations is $T = n \cdot \mathcal{T}$, and the logit index takes the form
$
X_{i\tau}' \beta + A_i + B_\tau,
$
where $A_i$ and $B_\tau$ are unit and time fixed effects, respectively.
Again, consider the linear model analogue:
\begin{align}
Y_{i\tau} = X_{i\tau}' \beta + A_i + B_\tau + \varepsilon_{i\tau}.
  \label{TwoWayPanelLinear}
\end{align}
For $n = \mathcal{T} = 2$, the standard difference-in-differences strategy eliminates both fixed effects:
$$
(Y_{11} - Y_{12}) - (Y_{21} - Y_{22}) = [(X_{11} - X_{12}) - (X_{21} - X_{22})]' \beta + [(\varepsilon_{11} - \varepsilon_{12}) - (\varepsilon_{21} - \varepsilon_{22})].
$$
With $Y=(Y_{11},Y_{12},Y_{21},Y_{22})'$, the linear combination in the last display corresponds to
$w_{\perp} = (1, -1, -1, 1)'$, which satisfies $w_{\perp} \in \{-1,0,1\}^T$ and is therefore applicable to the logit model as well.

For such $2 \times 2$ subpanels with two-way fixed effects, identification strategies based on  $w_{\perp} = (1, -1, -1, 1)'$ have been developed by \citet{charbonneau2017multiple} for binary logit models, and by \citet{jochmans2017two} for certain nonlinear models with multiplicative unobservables. Such model structures arise naturally in applications such as matched employer-employee data or international trade.

It is convenient to think of $w_{\perp}$ as the vectorization of an $n \times \mathcal{T}$ matrix. In the $n = \mathcal{T} = 2$ case:
$$
w_{\perp} = \mathrm{vec} \begin{pmatrix}
1 & -1 \\
-1 & 1
\end{pmatrix} = 
\begin{pmatrix}
1 \\ -1 \\ -1 \\ 1
\end{pmatrix}.
$$
Our general condition $\mathsf{W} w_{\perp} = 0$ is equivalent to requiring that this $n \times \mathcal{T}$ matrix has all row and column sums equal to zero.
When $n = \mathcal{T} = 3$, many valid vectors $w_{\perp}$ exist, for instance:
$$
w_{\perp} = \mathrm{vec} \begin{pmatrix}
1 & -1 & 0 \\
0 & 1 & -1 \\
-1 & 0 & 1
\end{pmatrix}.
$$
Thus, the construction generalizes to larger $n \times \mathcal{T}$ 
 subpanels.

\subsubsection{Dyadic network formation}
\label{ex:DyadicNetwork}

This corresponds to the structure in Example~\ref{ex:DyadicExample}. Recall that $T = n(n-1)/2$ and $d_w = n$, where each observation corresponds to an unordered dyad $(i,j)$ with $i \neq j$ and index structure
$ X_{ij}' \beta + A_i + A_j$.
The corresponding linear model reads
$$
Y_{ij} = X_{ij}' \beta + A_i + A_j + \varepsilon_{ij},
$$
which is essentially the same model as \eqref{TwoWayPanelLinear}, except for the symmetry
$Y_{ij} = Y_{ji}$ and that $Y_{ii}$
is unobserved.

For $n=3$ we cannot eliminate the fixed effects in this model, because we have the same number of observations $(Y_{12},Y_{13},Y_{23})$ as fixed effects $(A_1,A_2,A_3)$. However, for $n=4$
we can simply reproduce the same different strategy as in model \eqref{TwoWayPanelLinear} by considering the $2 \times 2$ subpanel given by $i \in \{1,2\}$ and $j \in \{3,4\}$, which gives
$$
(Y_{13} - Y_{14}) - (Y_{23} - Y_{24}) = [(X_{13} - X_{14}) - (X_{23} - X_{24})]' \beta + [(\varepsilon_{13} - \varepsilon_{14}) - (\varepsilon_{23} - \varepsilon_{24})].
$$
Defining the appropriate vectorization operator by
$$
(Y_{12},Y_{13},Y_{14},Y_{23},Y_{24},Y_{34})' = {\rm vech} \left(
\begin{array}{cccc}
* & Y_{12} & Y_{13} & Y_{14} \\
Y_{12} & * & Y_{23} & Y_{24} \\
Y_{13} & Y_{23} & * & Y_{34} \\
Y_{14} & Y_{24} & Y_{34} & *
\end{array}
\right),
$$
we can write the differencing vector used here 
as
$$
w_{\perp} = {\rm vech} \left(
\begin{array}{cccc}
* & 0 & 1 & -1 \\
0 & * & -1 & 1 \\
1 & -1 & * & 0 \\
-1 & 1 & 0 & *
\end{array}
\right),
$$
which satisfies  $w_{\perp} \in \{-1,0,1\}^T$, and is therefore equally applicable to the logit model. Indeed, the   tetrad configuration in \citet{graham2017econometric} corresponds 
exactly to this vector $w_{\perp}$ 
and the corresponding outcome pairs
satisfying $w_{\perp} = y_1 - y_2$. 

This idea extends to larger subnetworks. For instance, with $n = 5$, one finds valid examples such as:
$$
w_{\perp} = {\rm vech} \left(
\begin{array}{ccccc}
* & 1 & 0 & 0 & -1 \\
1 & * & 0 & 0 & -1 \\
0 & 0 & * & -1 & 1 \\
0 & 0 & -1 & * & 1 \\
-1 & -1 & 1 & 1 & *
\end{array}
\right),
\quad
w_{\perp} = {\rm vech} \left(
\begin{array}{ccccc}
* & 1 & 1 & -1 & -1 \\
1 & * & -1 & 1 & -1 \\
1 & -1 & * & -1 & 1 \\
-1 & 1 & -1 & * & 1 \\
-1 & -1 & 1 & 1 & *
\end{array}
\right).
$$
These $w_{\perp}$ again correspond to identifying configurations that satisfy the conditions of Theorem~\ref{th:IdentificationStatic}.

\subsubsection{Triadic network models}

This example generalizes the two-way panel model structure to three dimensions, following \citet{MurisPakel2025}. Each observation is indexed by a triad $t = (i,j,k)$, with $i \in \{1,\ldots,n_1\}$, $j \in \{1,\ldots,n_2\}$, and $k \in \{1,\ldots,n_3\}$ denoting elements from three disjoint partite sets. The logit index takes the form
$
X_{ijk}' \beta + A_{ij} + B_{jk} + C_{ik},
$
where $A_{ij}$, $B_{jk}$, and $C_{ik}$ are pairwise fixed effects.
Consider the linear model analogue:
$$
Y_{ijk} = X_{ijk}' \beta + A_{ij} + B_{jk} + C_{ik} + \varepsilon_{ijk}.
$$
To eliminate all three sets of pairwise fixed effects, we need a triple-differencing strategy. Consider a hexad consisting of two nodes from each part: $i \in \{i_1, i_2\}$, $j \in \{j_1, j_2\}$, and $k \in \{k_1, k_2\}$. The triple difference
$$[(Y_{i_1 j_1 k_1} - Y_{i_1 j_1 k_2}) - (Y_{i_1 j_2 k_1} - Y_{i_1 j_2 k_2})]   - [(Y_{i_2 j_1 k_1} - Y_{i_2 j_1 k_2}) - (Y_{i_2 j_2 k_1} - Y_{i_2 j_2 k_2})] 
$$
eliminates all pairwise fixed effects. This 
corresponds to $w_{\perp} = (1, -1, 1, -1, 1, -1, 1, -1)'$, and
since $w_{\perp} \in \{-1,0,1\}^T$, the same vector works for the logit model. Each pair appears exactly twice in the eight triads: once with weight $+1$ and once with weight $-1$, ensuring that the net effect for each $(i,j)$, $(j,k)$, and $(i,k)$ pair is zero, thus satisfying $\mathsf{W} w_{\perp} = 0$.

This hexad-based identification strategy for triadic models was developed by \citet{MurisPakel2025}. Such models arise naturally in applications where a ``team'' is formed by selecting one node from each of three sets—for instance, firm-industry-time interactions, or trade triplets involving importer, exporter, and product. Unlike in dyadic models, however, \citet{MurisPakel2025} show that even under sparsity, identification and asymptotic normality require stronger conditions: the presence of a growing number of informative hexads. Their framework illustrates how richer forms of unobserved heterogeneity, while modeling more realistic structures, come at a cost in terms of data requirements.

\subsection{Main Takeaways}

\begin{itemize}

\item
For the static binary logit models with fixed effects in Assumption~\ref{ass:StaticModel}, it is well known that conditioning on the sufficient statistic $\mathsf{W}Y$ removes dependence on the fixed effects $A$. However, identification of $\beta$ requires more than sufficiency: we need outcome pairs $(y_1, y_2)$ such that $\mathsf{W} y_1 = \mathsf{W} y_2$ and $X(y_1 - y_2)$ exhibits sufficient variation. This condition is central to Theorem~\ref{th:IdentificationStatic}.

\item
A key insight is the structural parallel to linear models (Remark~\ref{remark:LinearAnalogy}): Differencing strategies that eliminate fixed effects in linear models also eliminate them in binary logit models --- provided the differencing vector satisfies $w_{\perp} \in \{-1,0,1\}^T$. This constraint reflects the binary nature of the data and leads to a unifying framework for constructing valid conditional likelihood estimators across a wide range of panel and network models.

\end{itemize}

\section{Sufficient statistic for dynamic logit models}
\label{sec:DynamicsNoX}

\subsection{Model setup}

We now turn to dynamic binary choice logit models with generalized fixed effects, extending the approach used for static models in Section~\ref{sec:Static}. In the current section, we focus on purely dynamic specifications --- generalized autoregressive models --- without additional covariates $X$. As in the static case, our identification strategy relies on finding sufficient statistics that eliminate the fixed effects. In contrast, Section~\ref{sec:GeneralDynamic} below considers dynamic models that include covariates $X$, where identification will instead be based on moment conditions that are invariant to the fixed effects.

Most of the notation carries over from the static case. In particular, the definitions of $Y$, $W$, and $A$ remain unchanged. However, we require some additional structure to capture the temporal dependence in the outcomes.
Let $Y^0$ denote a vector of initial conditions, and define
$$
Y^{t-1} = (Y_{t-1}, Y_{t-2}, \ldots, Y_1, Y^0)
$$
to be the history of outcomes up to time $t-1$, including initial conditions. The following assumption gives the class of generalized autoregressive models considered in this section.

\begin{assumption} 
\label{ass:DynamicModelWithoutX}
The data-generating process is
\begin{align*}       
{\rm Pr} \left( Y=y \, \big| \, Y^0,A \right)  &= \prod_{t=1}^T
{\rm Pr} \left( Y_t=y_t \, \big| \, Y^{t-1},A \right) ,
\\
{\rm Pr} \left( Y_t=y_t \, \big| \, Y^{t-1},A \right)
&=  \frac {\left[\exp(\pi_t(Y^{t-1},\theta) + w_t' \, A)\right]^{y_t}} {1+\exp(\pi_t(Y^{t-1},\theta) + w_t' \, A) }   ,
\end{align*}
where $\theta$ is an unknown parameter
and $\pi_t(\cdot,\cdot)$ is a known
function for every $t \in \{1,\ldots,T\}$.
\end{assumption}

Note that the specification of the unobserved fixed effects $A \in \mathbb{R}^{d_w}$ and the non-random regressor vectors $w_t \in \mathbb{R}^{d_w}$ is unchanged from the static case. Crucially, the model imposes no restrictions on the dependence between the fixed effects $A$ and the initial condition $Y^0$.

\begin{example}[Dynamic panel data]
\label{ex:DynamicPanel}
This example generalizes Example~\ref{ex:StandardPanel} by allowing for state dependence through a lagged dependent variable. Specifically, suppose that 
$\pi_t(Y^{t-1}, \theta) =  Y_{t-1} \, \gamma,$
so that the model becomes:
$$
\Pr(Y_t = 1 \mid Y^{t-1},  A) = \frac{ \exp( Y_{t-1} \, \gamma   + A) }{ 1 + \exp( Y_{t-1} \, \gamma   + A) },
$$
where  $\gamma \in \mathbb{R}$ captures state dependence while $A \in \mathbb{R}$ captures
individual-specific heterogeneity as before.
For this baseline model, there are well-known
results due to \cite{Chamberlain78} and based on \citet{cox1958regression} that show that 
how to identify and estimate $\gamma$
via conditional likelihood by conditioning on sufficient statistics based on transition counts.\footnote{\cite{Chamberlain84}, \cite{magnac2004panel} and \cite{d2010duration} generalize those result further.}
Our goal in this section is to generalize those existing results to more general fixed effect structures.
\end{example}

\begin{example}[Dynamic dyadic network formation]
\label{ex:DynamicDyadic}
This example extends the dyadic model of Example~\ref{ex:DyadicExample} to a dynamic setting. Suppose we observe a sequence of undirected binary networks among $n$ agents over periods $\tau = 0, 1, \ldots, \mathcal{T}$. Each observation corresponds to a dyad $(i,j)$, and the link indicator $Y_{ij\tau} = Y_{ji\tau} \in \{0,1\}$ records whether a link between $i$ and $j$ exists in period $\tau$.
The dynamic logit model considered in \cite{graham2016homophily} takes the form:
\[
\Pr(Y_{ij\tau} = 1 \mid Y_{ij,\tau-1}, R_{ij,\tau-1}, A_{ij}) = 
\frac{ \exp\left(  Y_{ij,\tau-1} \, \gamma +  R_{ij,\tau-1} \, \delta + A_{ij} \right) }
{ 1 + \exp\left(  Y_{ij,\tau-1} \, \gamma +  R_{ij,\tau-1} \, \delta + A_{ij} \right) },
\]
where $ Y_{ij,\tau-1}$ is the lagged link indicator
for the same dyad, $R_{ij,\tau-1} := \sum_{k \notin \{i,j\}} Y_{ik,\tau-1} Y_{jk,\tau-1}$ is the number of shared friends of $i$ and $j$ in the previous period, $A_{ij} \in \mathbb{R}$ is a dyad-specific fixed effect,
and $\gamma$ and $\delta$ are unknown parameters.
To map this into our general framework (Assumption~\ref{ass:DynamicModelWithoutX}), we define $T = \binom{n}{2} \cdot \mathcal{T}$, where $\mathcal{T}$ is the number of observed time periods in Graham’s model, and $\binom{n}{2}$ is the number of dyads. We then treat each dyad-time pair $(i,j,\tau)$ as a single observation indexed by $t = 1, \ldots, T$, see Section~\ref{subsec:DynamicNetworkNoX} below for more details.
In this formulation, the logit index for observation $t$ becomes
$\pi_t(Y^{t-1}, X\theta) =  Y_{ij,\tau-1} \, \gamma +  R_{ij,\tau-1} \, \delta$.
\end{example}

\subsection{Generalized sufficient statistics}

The following lemma is central to all identification results in this section.

\begin{lemma}
    \label{lemma:GeneralDynamicsSufficientStats}
    Suppose Assumption~\ref{ass:DynamicModelWithoutX} holds with fixed initial conditions $Y^0 = y^0$.
Let $y, \tilde y \in \{0,1\}^T$ be any two outcome sequences, and define
$y^{t-1} = (y_{t-1}, y_{t-2}, \ldots, y_1, y^0)$ and
$\tilde y^{\, t-1} = (\tilde y_{t-1}, \tilde y_{t-2}, \ldots, \tilde y_1, y^0)$, for each $t \in \{2,\ldots,T\}$.
Assume further that:
    \begin{itemize}
       \item[(i)]
    $\sum_{t=1}^T w_t y_t = \sum_{t=1}^T w_t \tilde{y}_t$.

    \item[(ii)]  
    $\Big[ \big(w_t,\pi_t(y^{t-1},\theta)\big) \, : \, t=2,\ldots,T\} \Big]$  
    is a permutation of  
    $\Big[ \big(w_t,\pi_t(\tilde y^{t-1},\theta)\big) \, : \, t=2,\ldots,T\} \Big].$
    \end{itemize}
    Then we have
$$
\frac{{\rm Pr}(Y=y | Y^0=y^0, A)}{{\rm Pr}(Y=\tilde y | Y^0=y^0, A)} = \exp\left\{   \sum_{t=1}^{T}\left[ y_{t} \, \pi_t(y^{t-1},\theta) - \tilde y_{t} \, \pi_t(\tilde y^{t-1},\theta) \right] \right\} ,
$$
which does not depend on $A$.
\end{lemma}

The proof is given in the appendix.

\subsection{AR($p$) panel data models}

Our first example of a model that satisfies Assumption~\ref{ass:DynamicModelWithoutX} and for which Lemma~\ref{lemma:GeneralDynamicsSufficientStats} yields immediate and useful identification results is the AR(1) model. In this case, we take $\theta = \gamma \in \mathbb{R}$, and set
$
\pi_t(Y^{t-1}, \theta) = Y_{t-1} \, \gamma,
$
with initial condition $Y^0 = Y_0$ given by the observed outcome in period $t = 0$.
Under this specification, the model in Assumption~\ref{ass:DynamicModelWithoutX} becomes:
\begin{align}       
{\rm Pr} \left( Y=y \, \big| \, Y_0=y_0,A \right)  &=  \prod_{t=1}^T \, 
\left[ \frac {1} {1+\exp(y_{t-1} \, \gamma + w_t' \, A) }  
\right]^{1-y_t}
\left[ \frac {\exp(y_{t-1} \, \gamma + w_t' \, A)} {1+\exp(y_{t-1} \, \gamma + w_t' \, A) }  
\right]^{y_t} .
   \label{ModelAR1}
\end{align}
We are going to discuss the AR(1) model in detail here, and afterwards briefly summarize the extension to AR($p$) model for $p>1$, with details given in the appendix.

In order to obtain identification results that mirror the results for the static model above as close as possible, we furthermore impose the following assumption on $w_t = (w_{t,1}, \ldots, w_{t,d_w})$ here:
\begin{align}
   w_{t,k} \in \{0,1\}, \quad
   \text{for all $k \in \{1,\ldots,d_w\}$}, \quad \text{and} \quad \sum_{k=1}^{d_w} w_{t,k} =1, \quad
   \text{for all $t \in \{1,\ldots,T\}$.}
   \label{RestrictionWbinary}
\end{align}
{\color{red} }These constraints on $w_t$  may appear overly restrictive at first glance. However, as we will explain in Remark~\ref{remark:WtRestriction}, these constraints are in fact without loss of generality from the perspective of constructing sufficient statistics for the fixed effect $A$ in this model.
To simplify notation, we define the lagged outcome $T$-vector as $Y_{\rm lag} = (Y_0, Y_1, \ldots, Y_{T-1})'.$
Similarly, we define the ``lead'' version of $\mathsf{W} = (w_1, \ldots, w_T)$ by $\mathsf{W}_{\rm lead} = (w_2, w_3,\ldots, w_T,0_{d_w \times 1}) \in \mathbb{R}^{d_w \times T}$.

\begin{theorem}
\label{th:AR1Identification}
We assume the AR(1) model in \eqref{ModelAR1}, with $w_t$ satisfying the restrictions in \eqref{RestrictionWbinary}.
We treat the initial condition $Y_0 = y_0$ as fixed and known.
    \begin{enumerate}[(i)]
       \item Then, $(\mathsf{W}Y, \mathsf{W}Y_{\rm lag})$ is a  statistic that is sufficient for the nuisance parameter $A$ (conditional on $Y_0$), in the sense that the distribution of $Y$ given $(\mathsf{W}Y, \mathsf{W}Y_{\rm lag})$ and $Y_0$ does not depend on $A$.
      Since $Y_0 = y_0$ is treated as fixed, we can equivalently express this sufficient statistic as the linear transformation $\mathsf{V} Y$, where 
       $\mathsf{V} := (\mathsf{W}'  ,\mathsf{W}'_{\rm lead})'$.

        \item Suppose further that there exist two binary vectors $y, \tilde y \in \{0,1\}^T$ such that
    $\mathsf{V} y  = \mathsf{V} \tilde y $
    and 
    $\sum_{t=1}^{T} y_{t} y_{t-1} \neq \sum_{t=1}^T \tilde y_{t} \tilde y_{t-1}$,
    where we set $Y_0=y_0 = \tilde y_0$. Then, the autoregressive parameter $\gamma$ is uniquely identified from the distribution of $Y$ conditional on $Y_0 = y_0$.
 
    \end{enumerate}
\end{theorem}

The proof is provided in the appendix. Part (i) of the theorem establishes that the sufficient statistic for the fixed effect $A$ in this model is given by the pair $(\mathsf{W}Y, \mathsf{W}Y_{\rm lag})$, which can be explicitly written as 
$\left(\sum_{t=1}^T w_t y_t,\sum_{t=1}^T w_t y_{t-1}\right)$. The first term, $\sum_{t=1}^T w_t y_t$, is familiar from the static model in Section~\ref{sec:Static}, where it formed the sufficient statistic for $A$. The second term, $\sum_{t=1}^T w_t y_{t-1}$, is new and arises due to the autoregressive structure.

To understand the role of these statistics, consider two outcome paths $y, \tilde y \in \{0,1\}^T$. Using only $\sum_{t=1}^T w_t y_t = \sum_{t=1}^T w_t \tilde{y}_t$
and $y_{0} = \widetilde y_0$, the likelihood ratio simplifies as follows:
$$
\frac{{\rm Pr}(Y=y | Y_0=y_0, A)}{{\rm Pr}(Y=\tilde y | Y_0=y_0, A)} = \exp\left[\gamma  \sum_{t=1}^{T}(y_{t} y_{t-1} - \tilde y_{t} \tilde y_{t-1}) \right] \prod_{t=2}^T \frac{1+\exp( \tilde y_{t-1} \gamma+ w_t' A)}{1+\exp( y_{t-1}  \gamma+ w_t' A)} .
$$
Here, the second term still depends on $A$, which shows that, in contrast to the static model, matching only $\sum_{t=1}^T w_t y_t$ is not sufficient to eliminate the fixed effect from the likelihood ratio.
To fully eliminate $A$ from this ratio, we also need to use that
$\sum_{t=2}^T y_{t-1} = \sum_{t=2}^T \tilde y_{t-1}$ and
$\sum_{t=2}^T w_t y_{t-1} = \sum_{t=2}^T w_t \tilde{y}_{t-1}$ and $w_{t,k} \in \{0,1\}$ with  $\sum_{k=1}^{d_w} w_{t,k} =1$. Under those conditions, one can show that
$\prod_{t=2}^T \frac{1+\exp(\gamma \tilde y_{t-1} + w_t' A)}{1+\exp(\gamma y_{t-1} + w_t' A)}=1$, and therefore
$$
\frac{{\rm Pr}(Y=y | Y_0=y_0, A)}{{\rm Pr}(Y=\tilde y | Y_0=y_0, A)} = \exp\left[\gamma  \sum_{t=1}^{T}(y_{t} y_{t-1} - \tilde y_{t} \tilde y_{t-1}) \right]  
$$
This shows that $\gamma$ can be point-identified from the likelihood ratio,
as long as $\sum_{t=1}^{T} y_{t} y_{t-1} \neq \sum_{t=1}^T \tilde y_{t} \tilde y_{t-1}$, which is precisely the content of part (ii) of the theorem.

\begin{remark}
\label{remark:WtRestriction}
The restrictions $w_{t,k} \in \{0,1\}$ and $\sum_{k=1}^{d_w} w_{t,k} \in \{0,1\}$ in Assumption~\ref{RestrictionWbinary} may appear overly restrictive. However, these conditions are without loss of generality for the purpose of constructing sufficient statistics and identifying $\gamma$ via conditional likelihood methods.
To see this, consider the general autoregressive model where $w_t \in \mathbb{R}^{d_w}$ can take any values. Let $\Omega = \{\varphi_1, \varphi_2, \ldots, \varphi_{d_\omega}\} \subset \mathbb{R}^{d_w}$ denote the set of distinct values that $w_t$ assumes across $t = 1, \ldots, T$, where
$d_\omega$ is the cardinality of $\Omega$.
We can then define indicator variables
$\omega_{t,k} = \mathbbm{1}(w_t = \varphi_k)$, for $k = 1, \ldots, d_\omega$,
implying that $\omega_t = (\omega_{t,1}, \ldots, \omega_{t,d_\omega})'$ satisfies the binary restrictions on $w_t$ in \eqref{RestrictionWbinary}.
We then have $w_t' A = \omega_t' A^*$, where $A^* = (\varphi_1' A, \ldots, \varphi_{d_\omega}' A)' \in \mathbb{R}^{d_\omega}$. 

One can show that the sufficient statistics for $A$ in the original model with $w_t' A$ are identical to those for $A^*$ in the transformed model with $\omega_t' A^*$, as characterized in Theorem~\ref{th:AR1Identification}.
\end{remark}

\begin{remark}
   Applying Lemma~\ref{lemma:Wperp} in the context of part (ii)
   of Theorem~\ref{th:AR1Identification} gives the following:
       There exist  $y,\tilde y \in \{0,1\}^T$ with
   $y \neq \tilde y$ and
    $\mathsf{V} y_1  = \mathsf{V} y_2  $
  if and only if  
   there exists $w_{\perp} \in \{-1,0,1\}^T$
   with $w_{\perp} \neq 0$ and  $\mathsf{V} w_{\perp} = 0$.
   This can be computationally very useful to construct such pairs $y,\tilde y \in \{0,1\}^T$ that allow identification of $\gamma$. 
However, the additional identification condition  $\sum_{t=1}^{T} y_{t} y_{t-1} \neq \sum_{t=1}^T \tilde y_{t} \tilde y_{t-1}$ is non-linear in the outcomes
and can therefore not be expressed in terms of $w_{\perp}$ only.   
\end{remark}

\newenvironment{examplecont}[1]{%
  \refstepcounter{example}%
  \renewcommand{\theexample}{\ref{#1} (cont.)}%
  \begin{example}%
}{%
  \end{example}%
}

\begin{examplecont}{ex:DynamicPanel}
The simplest example covered by Theorem~\ref{th:AR1Identification} is the case where $w_t = 1$ for all $t = 1, \ldots, T$, implying a standard individual-specific fixed effect $A$ that enters identically in every period. In this setting, the model reduces to a pure logit AR(1) panel model with a scalar fixed effect and no covariates. The sufficient statistics for $A$ in this model — namely, $(\sum_{t=1}^T Y_t, \sum_{t=1}^T Y_{t-1})$ — are well known since \citet{cox1958regression}, and are often expressed in terms of $Y_0$ (which we condition on throughout), $\sum_{t=1}^{T-1} Y_t$, and $Y_T$.
To identify the autoregressive parameter $\gamma$, one must observe at least $T = 3$ periods (in addition to the initial condition $Y_0$). For example, when $T=3$ and $Y_0 = 0$, the sequences $y = (1, 0, 1)$ and $\tilde{y} = (0, 1, 1)$ satisfy all conditions in part (ii) of Theorem~\ref{th:AR1Identification} to guarantee identification of $\gamma$.
\end{examplecont}

\begin{example} \label{ExampleQuarterlyFixedEffects}
Consider quarterly panel data on binary outcomes $Y_0, Y_1, \ldots, Y_T$, and suppose we specify a logit AR(1) model with a single index of the form
$$
Y_{t-1} \, \gamma + \sum_{q=1}^4 \, A_{q} \, w_{t,q} ,
$$
where $w_{t,q}$ is an indicator for quarter $q \in \{1,2,3,4\}$
and $A_{q}$ are quarter and unit-specific fixed effects.
To identify $\gamma$ based on part (ii) of Theorem~\ref{th:AR1Identification} we require $y, \tilde y \in \{0,1\}^T$ that satisfy 
$$
\sum_{t=1}^T w_{qt}(y_{t} - \tilde y_{t}) = 0 ,
\quad \text{and} \quad
\sum_{t=1}^T w_{qt}(y_{t-1} - \tilde y_{t-1}) = 0,
\quad \text{for each } q = 1,2,3,4.
$$
Finding a solution with $y \neq \tilde y$ requires $T = 6$ quarterly observations (plus the initial quarter $Y_0$),
and the solution satisfies
$y - \tilde y = \pm  (1, 0, 0, 0, -1, 0)$.
In addition, we need to satisfy the condition $\sum_{t=1}^{T} y_{t} y_{t-1} \neq \sum_{t=1}^T \tilde y_{t} \tilde y_{t-1}$,
e.g.\ for $Y_0=0$, by choosing $y = (0,0,0,0,1,1)$
and $\tilde y = (1,0,0,0,0,1)$. This shows that $\gamma$ is identified for $T=6$.
\end{example}

\begin{example}\label{ex:Heterogeneous Time Trend}
Consider the case $d_w = 2$ with $w_{t,1} = 1$ and $w_{t,2} = t$, corresponding to a model with heterogeneous linear time trends and single index
\[
Y_{t-1} \, \gamma + A_1 + t \, A_2.
\]
By Remark~\ref{remark:WtRestriction}, this model can be reparameterized — for the purpose of constructing sufficient statistics — as a model with time-specific fixed effects:
\[
Y_{t-1} \, \gamma + A_t,
\]
where $d_w = T$ and $A_t$ denotes a distinct fixed effect for each time period. In this formulation, the fixed effects vary freely over time, and no nontrivial sufficient statistics can be constructed to eliminate $A$. As a result, our sufficiency-based identification strategy cannot be applied, and Theorem~\ref{th:AR1Identification} yields a negative result in this case: the autoregressive parameter $\gamma$ is not identified via conditional likelihood.
However, as discussed in Section~4.3 of \citet{honore2024dynamic}, identification is still possible using an alternative approach. Specifically, in the autoregressive panel model with heterogeneous linear time trends (i.e., index $Y_{t-1} \, \gamma + A_1 + t \, A_2$), the parameter $\gamma$ can be identified via a moment condition strategy, as described in Section~\ref{sec:GeneralDynamic} below, provided that $T \geq 8$.
This example illustrates that for sufficiently general index structures — particularly those involving unrestricted time-varying fixed effects — the sufficient statistics approach may fail, while the moment condition approach can still succeed.
\end{example}

\subsubsection*{Generalization to AR($p$) models with $p>1$}

We now generalize the previous example by considering an AR(2) model with $\pi_t(Y^{t-1}, \theta) = Y_{t-1} \, \gamma_1+ Y_{t-2} \, \gamma_2$ and initial condition $Y^0 = (Y_0,Y_{-1})\in\{0,1\}^2$. Under this setup, the model specified in Assumption~\ref{ass:DynamicModelWithoutX} becomes: 
\begin{align}       
&{\rm Pr} \left( Y=y \, \big| \, Y_0=y_0,A \right) \notag
\\
&=  \prod_{t=1}^T \, 
\left[ \frac {1} {1+\exp(y_{t-1} \, \gamma_1 + y_{t-2} \, \gamma_2 + w_t' \, A) }  
\right]^{1-y_t}
\left[ \frac {\exp(y_{t-1} \, \gamma_1 + y_{t-2} \, \gamma_2 + w_t' \, A)} {1+\exp(y_{t-1} \, \gamma_1 + y_{t-2} \, \gamma_2 + w_t' \, A) }  
\right]^{y_t} . \label{ModelAR2}
\end{align}
We refer readers to Appendix Section \ref{Details_ARp} for a general treatment of the AR($p$) models with $p>1$. Throughout, we maintain the assumption that for all $t \in \{1, \ldots, T\}$, the vector of weights $w_t = (w_{t,1}, \ldots, w_{t,d_w})$ satisfies the restrictions in \eqref{RestrictionWbinary}.

Applying Lemma \ref{lemma:GeneralDynamicsSufficientStats}, we show in Appendix Section \ref{Details_ARp} that the likelihood ratio for two outcome paths $y,\tilde y\in \{0,1\}^T$ \footnote{These generalize the sufficient statistics found in \cite{Chamberlain84} and \cite{magnac2000subsidised}.}
$$\frac{{\rm Pr}(Y=y | Y_0=y_0, A)}{{\rm Pr}(Y=\tilde y | Y_0=y_0, A)} $$
is invariant to $A$ if the following four conditions hold:
\begin{align} \label{SufficientStats_pureAR2}
   \sum_{t=1}^T w_t \, y_t &= \sum_{t=1}^T w_t \, \tilde{y}_t 
   \\
   \sum_{t=1}^T w_t \, y_{t-1} &= \sum_{t=1}^T w_t \, \tilde{y}_{t-1} \notag
   \\
   \sum_{t=1}^T w_t \, y_{t-2} &= \sum_{t=1}^T w_t \, \tilde{y}_{t-2} \notag
   \\
   \sum_{t=1}^T w_t \, y_{t-1} \, y_{t-2} &= \sum_{t=1}^T w_t \, \tilde y_{t-1} \, \tilde{y}_{t-2}. \notag
\end{align}
 The first condition coincides with Assumption (i) in Lemma \ref{lemma:GeneralDynamicsSufficientStats} while the remaining three are equivalent statements of Assumption (ii) for the AR(2). That is, they ensure together that $\Big[ \big(w_t,y_{t-1}\, \gamma_1+y_{t-2}\, \gamma_2\big) \, : \, t=2,\ldots,T\Big] $ 
is a permutation of  
$\Big[ \big(w_t,\tilde y_{t-1}\, \gamma_1+\tilde y_{t-2}\, \gamma_2\big) \, : \, t=2,\ldots,T\Big]$.

 \begin{example}
Consider the standard AR(2) model with $d_w = 1$, $w_{t} = 1$ for all $t$. In this case, the conditions in \eqref{SufficientStats_pureAR2} can only be satisfied if $T\geq 4$. For instance, with $T=4$ and initial condition $Y^0=(0,1)$, $y=(1,0,0,0)$ and $\tilde y=(0,1,0,0)$ 
form a valid pair. See also \cite{Chamberlain_1985} and \cite{honore2000panel}.
 \end{example}

 \begin{example}
In the spirit of Example \ref{ExampleQuarterlyFixedEffects}, consider quarterly panel data on binary outcomes $Y_{-1},Y_0, Y_1, \ldots, Y_T$, and assume an AR(2) structure with a single index 
$$
Y_{t-1} \, \gamma_1+Y_{t-2} \, \gamma_2 + \sum_{q=1}^4 \, A_{q} \, w_{t,q} ,
$$
where $w_{t,q}$ are indicators for quarter $q \in \{1,2,3,4\}$. In this setting, the set of restrictions \eqref{SufficientStats_pureAR2} can only be satisfied for $T\geq 7$. For example, for the initial condition $Y^0=(0,0)$, the outcome sequences $y=(0,0,0,0,1,0,1)$ and $\tilde y = (1,0,0,0,0,0,1)$ form a pair that allows  identification of $\gamma_2$.
 \end{example}

 \begin{remark}
 \label{Remark:ARp:conditionalMLE}
In general, the conditional likelihood approach does not allow for the identification of parameters other than the final autoregressive coefficient $\gamma_p$ in the (pure) AR($p$) model. To build intuition, consider again the simple case $p=2$. Then, it is a straightforward exercise (see e.g Appendix Section \ref{Details_ARp}) to show that the restrictions in \eqref{SufficientStats_pureAR2} imply $\sum_{t=1}^T y_{t}y_{t-1}=\sum_{t=1}^T \tilde{y}_{t}\tilde{y}_{t-1}$. In turn, this causes $\gamma_1$ to drop out from the likelihood ratio  $\frac{{\rm Pr}(Y=y | Y_0=y_0, A)}{{\rm Pr}(Y=\tilde y | Y_0=y_0, A)}$. This example illustrates the property that the sufficient statistics resulting from Lemma \ref{lemma:GeneralDynamicsSufficientStats} absorb all variation relevant for identifying $\gamma_1,\ldots,\gamma_{p-1}$. This leaves $\gamma_{p}$ as the only parameter that can be identified via conditional likelihood in the AR($p$) setting.
However, as was the case for the model with heterogeneous time trend in Example \ref{ex:Heterogeneous Time Trend}, identification is possible using an alternative approach that relies on moment conditions. 
\end{remark}

\subsection{Dynamic dyadic network formation}
\label{subsec:DynamicNetworkNoX}

We now consider the dyadic network formation model of \cite{graham2016homophily}, introduced in Example~\ref{ex:DynamicDyadic}. This model falls within the class of generalized autoregressive logit models studied above, and we show how our identification strategy based on Lemma~\ref{lemma:GeneralDynamicsSufficientStats} applies in this setting.

We model binary outcomes $Y_{ij\tau} = Y_{ji\tau} \in \{0,1\}$ for individuals $i,j \in \{1,\ldots,n\}$ with $i \neq j$, and time periods $\tau \in \{1,\ldots,{\cal T}\}$. The initial condition at $\tau=0$ is denoted by $Y^0$. Each dyad-time pair $(i,j,\tau)$ defines a single observation, and we set $T = \binom{n}{2} \cdot \mathcal{T}$ as the total number of observations. Each observation $t = (i,j,\tau)$ is treated as an element of the outcome vector $Y = (Y_1,\ldots,Y_T)' \in \{0,1\}^T$.

To express the model in the form of Assumption~\ref{ass:DynamicModelWithoutX}, we impose a deterministic ordering of dyads within each time period. The overall observation index $t$ respects time ordering, i.e.\ if $\tau_1 < \tau_2$ then all observations from period $\tau_1$ precede those from $\tau_2$. For each observation $t = (i,j,\tau)$, we define the regressor $w_t \in \mathbb{R}^{\binom{n}{2}}$ as a unit vector selecting the dyad $(i,j)$, so that $w_t' A = A_{ij}$, where $A \in \mathbb{R}^{\binom{n}{2}}$ collects all dyad-specific fixed effects. The logit index for observation $t$ is then
\[
\pi_t(Y^{t-1}, \theta) + w_t' A = Y_{ij,\tau-1} \, \gamma + R_{ij,\tau-1} \, \delta + A_{ij},
\]
with parameters $\theta = (\gamma,\delta)'$ and $R_{ij,\tau-1} = \sum_{k \notin \{i,j\}} Y_{ik,\tau-1} Y_{jk,\tau-1}$ denoting the number of shared friends in the previous period.

Throughout, we write $Y_{\cdot \cdot ,\tau}$ for the collection of $Y_{ij,\tau}$ over all dyads, and use ${\cal D}_n$ to denote the set of all ${n \choose 2}$ dyads. Given $Y_{\cdot \cdot ,\tau}$ and $(i,j) \in {\cal D}_n$, the covariates entering the index at time $\tau+1$ are
\[
Z_{ij}(Y_{\cdot \cdot,\tau})
:= \left( Y_{ij,\tau}, \sum_{k \neq i,j} Y_{ik,\tau} Y_{jk,\tau} \right) \in \{0,1\} \times \mathbb{Z}_+.
\]
We now describe the sufficient statistic structure implied by Lemma~\ref{lemma:GeneralDynamicsSufficientStats}. To simplify the exposition, we focus on the minimal case ${\cal T} = 3$, which is also the specification considered in \cite{graham2016homophily}. For each outcome vector $y \in \{0,1\}^T$, we define the set of alternative sequences
\begin{align*}
    {\cal Y}_{\rm cond}(y) = \bigg\{ 
      \tilde y \in \{0,1\}^T \, : \, &  \forall (i,j) \in {\cal D}_n
     \;  \text{we have: } y_{ij,3} = \tilde y_{ij,3}, \\
     &\text{and for all } \tau \in \{1,2\}, \text{ we have: } \\
     &\qquad Z_{ij}(y_{\cdot \cdot,\tau}) = Z_{ij}(\tilde y_{\cdot \cdot,\tau})
     \text{ or } 
     Z_{ij}(y_{\cdot \cdot,\tau}) = Z_{ij}(\tilde y_{\cdot \cdot,3-\tau})
    \bigg\}.
\end{align*}
By Lemma~\ref{lemma:GeneralDynamicsSufficientStats}, the conditional likelihood ratio
${\rm Pr}(Y=y \mid Y^0, A) / {\rm Pr}(Y=\tilde y \mid Y^0, A)$ is invariant to $A$ for all $\tilde y \in {\cal Y}_{\rm cond}(y)$.
This yields the conditional likelihood
\[
\mathcal{L}_{\text{cond}}(\gamma, \delta) = \frac{\prod_{(i,j) \in \mathcal{D}_n} 
\exp\left( \sum_{\tau=2}^{3} Y_{ij,\tau} \, (\gamma Y_{ij,\tau-1} + \delta R_{ij,\tau-1}) \right)
}{
\sum_{\tilde y \in   {\cal Y}_{\rm cond}(Y)} \prod_{(i,j) \in \mathcal{D}_n} 
\exp\left( \sum_{\tau=2}^{3} \tilde{y}_{ij,\tau} \, (\gamma \tilde{y}_{ij,\tau-1} + \delta \tilde R_{ij,\tau-1}) \right)
},
\]
where $\tilde R_{ij,\tau-1} =  \sum_{k \notin \{i,j\}} \tilde{y}_{ik,\tau-1} \cdot \tilde{y}_{jk,\tau-1}$. While this formulation is exact, the denominator is computationally infeasible to evaluate for large $n$, as the set ${\cal Y}_{\rm cond}(y)$ grows exponentially with the number of dyads.

To address this, \cite{graham2016homophily} proposes a restricted subset of ${\cal Y}_{\rm cond}(y)$ based on the concept of ``stable dyads'', which allows feasible computation in large-network settings. An alternative approach, suitable for small $n$ but many independent network observations (as in \citealt{graham2013comment}), is to compute the exact conditional likelihood using the full set ${\cal Y}_{\rm cond}(y)$ or a tractable subset.

A particularly convenient and computationally efficient subset is the two-element set
\begin{align*}
    {\cal Y}^*_{\rm cond}(y) = \bigg\{ 
      \tilde y \in \{0,1\}^T \, : \, &   
        y_{\cdot\cdot,3} = \tilde y_{\cdot\cdot,3}, \\
     &\text{and for all } \tau \in \{1,2\}: \; 
      y_{\cdot\cdot,\tau} = \tilde y_{\cdot\cdot,\tau}
      \text{ or }
         y_{\cdot\cdot,\tau} = \tilde y_{\cdot\cdot,3-\tau} 
    \bigg\}.
\end{align*}
This set contains at most two elements: the original network $y$ and the version with periods $\tau=1$ and $2$ swapped. When $y_{\cdot\cdot,1} = y_{\cdot\cdot,2}$, the set is a singleton. For $n=3$, one can show that ${\cal Y}^*_{\rm cond}(y) = {\cal Y}_{\rm cond}(y)$ for all networks $y$. For $n=4$ and $n=5$, numerical checks confirm that ${\cal Y}^*_{\rm cond}(y) = {\cal Y}_{\rm cond}(y)$ in more than 95\% of all configurations, making this approximation attractive in small-network applications.

In summary, this framework supports two estimation strategies:
Firstly, in settings with large $n$ and few time periods, feasible estimation can be achieved using restricted conditioning sets such as those based on stable dyads (\citealp{graham2016homophily}).
Secondly, in settings with small $n$ but many repeated network observations, exact conditional likelihood estimation is tractable and efficient using either ${\cal Y}_{\rm cond}(y)$ or its simplification ${\cal Y}^*_{\rm cond}(y)$.

The results extend naturally to longer time horizons ${\cal T} > 3$, though the structure of the conditioning sets becomes more complex. Likewise, the framework accommodates richer dynamic specifications, provided the conditions in Lemma~\ref{lemma:GeneralDynamicsSufficientStats} continue to hold.

\subsection{Main Takeaways}

\begin{itemize}

   \item The conditional likelihood approach based on sufficient statistics extends naturally to dynamic models. Lemma~\ref{lemma:GeneralDynamicsSufficientStats} characterizes the sufficient statistics for generalized autoregressive logit models with fixed effects. We have shown how to apply this result in both dynamic panel and dynamic network settings.

   \item  In contrast to the static case, the sufficient statistics approach does not always identify all the common parameters in dynamic models. We illustrated this in two examples  (i) in the AR(1) model with heterogeneous time trends, the sufficient statistics approach fails to identify the autoregressive parameter $\gamma$; (ii) in AR($p$) models with $p > 1$, the sufficient statistics approach can identify only the final lag coefficient $\gamma_p$. In the next section, we discuss an alternative to the sufficient conditions approach, which is based on moment conditions. This approach will lead to identification of all the common parameters in the two examples.

\end{itemize}

\section{Moment conditions for dynamic logit models}
 \label{sec:GeneralDynamic}

 We now turn to an alternative identification strategy for non-linear panel models based on moment conditions that do not depend on the fixed effects. Unlike the conditional likelihood approach used in Sections~\ref{sec:Static} and~\ref{sec:DynamicsNoX}, which relies on sufficient statistics, the moment-based approach developed here applies to a larger class of models and accommodates richer forms of unobserved heterogeneity. 

While fixed-effect-free moment conditions had been constructed in specific models before, the general idea was formalized in the context of semiparametric panel models by \citet{bonhomme2012functional}, who called it ``functional differencing''. Its applicability to discrete choice models was recently emphasized by \citet{kitazawa2021transformations} and \citet{honore2024dynamic}.

\subsection{Model setup}

The model and notation are essentially unchanged
compared to Section~\ref{sec:DynamicsNoX}, the only difference is that we now also allow for additional strictly exogenous covariates $X$,
as we had in Section \ref{sec:Static}.
Assumption~\ref{ass:DynamicModelWithoutX} then generalizes as follows.

\begin{assumption} 
\label{ass:DynamicModel}
The data-generating process is
\begin{align*}       
{\rm Pr} \left( Y=y \, \big| \, Y^0,X,A \right)  &= \prod_{t=1}^T
{\rm Pr} \left( Y_t=y_t \, \big| \, Y^{t-1},X,A \right) ,
\\
{\rm Pr} \left( Y_t=y_t \, \big| \, Y^{t-1},X,A \right)
&=  \frac {\left[\exp(\pi_t(Y^{t-1},X_t,\theta) + w_t' \, A)\right]^{y_t}} {1+\exp(\pi_t(Y^{t-1},X_t,\theta) + w_t' \, A) }   ,
\end{align*}
where $\theta$ is an unknown parameter
and $\pi_t(\cdot,\cdot,\cdot)$ is a known
function for every $t \in \{1,\ldots,T\}$.
\end{assumption}

\medskip
Our goal is again to identify the parameter $\theta$ without imposing restrictions on the distribution of the fixed effect $A$. 
The model in Assumption~\ref{ass:DynamicModel} nests many dynamic panel and network models.
Our baseline Examples~\ref{ex:DynamicPanel}
and~\ref{ex:DynamicDyadic} remain essentially unchanged, except for the additional covariates $X_t$. Thus, in Example~\ref{ex:DynamicPanel} the specification becomes
$$
\pi_t(Y^{t-1}, X_t, \theta) = Y_{t-1} \, \gamma + X_t' \, \beta,
$$
where $\theta=(\gamma,\beta)'$
while generalizing Example~\ref{ex:DynamicDyadic} gives
$$
\pi_t(Y^{t-1}, X_t, \theta) =  Y_{ij,\tau-1} \, \gamma +  R_{ij,\tau-1} \, \delta + X_t' \, \beta
$$
with $\theta = (\gamma, \delta)'$.

It is sometimes possible to extend the identification strategy based on sufficient statistics to models with additional covariates. For example, \citet{honore2000panel} apply this approach to an AR(1) panel logit model with covariates by conditioning on subsets of the data where covariate values are identical across adjacent time periods. More generally, the results in Section~\ref{sec:DynamicsNoX} can be applied to such models by treating $X$ as fixed throughout and absorbing the covariate dependence of $\pi_t(Y^{t-1}, X_t, \theta)$ into the definition of $\pi_t$. However, the sufficiency-based approach typically works only under restrictive support conditions on the covariates. In contrast, the moment condition strategy developed in this section applies more broadly and often accommodates arbitrary variation in $X$.

\subsection{Identification via fixed-effect-free moment conditions}

As discussed above, the sufficient statistics approach typically breaks down in dynamic settings with covariates that vary freely over time. This motivates an alternative strategy based on moment functions that depend on the parameters of interest but not on the fixed effects.
Specifically, we aim to find functions $m(Y, Y^0, X, \theta)$ such that
\begin{align}
   \mathbb{E}\left[
   m(Y, Y^0,X,\theta)
   \, \big| \,  Y^0,X,A \right] = 0.   
   \label{MomentFunctions}
\end{align}
Moment functions satisfying \eqref{MomentFunctions} are valid for identification and estimation, as they are invariant to the fixed effect $A$ and therefore unaffected by the incidental parameter problem. To characterize when such functions exist, we use a combinatorial argument that provides a lower bound on the dimension of the space of fixed-effect-free moment functions,\footnote{See also \cite{dano2023transition} for an alternative strategy tailored to standard AR($p$) logit models with $p \geq 1$.} leading to the following theorem.

To state the theorem, we introduce some additional notation.
For each time period $t \in \{1, \ldots, T\}$, define the set of distinct values that the index $\pi_t(Y^{t-1}, X_t, \theta)$ can take across all possible outcome histories $(y_{t-1},y_{t-2},\ldots,y_1) \in \{0,1\}^{t-1}$  as
$$
  \Pi_t( Y^0,X,\theta) := \big\{ \pi_t(( y_{t-1},y_{t-2},\ldots,y_1,Y^{0}),X_t,\theta) \, :
   \, (y_{t-1},y_{t-2},\ldots,y_1) \in \{0,1\}^{t-1} \big\} ,
$$
Throughout this section, we treat $Y^0$, $X$, and $\theta$ as fixed. While we make these arguments explicit in the notation for mathematical precision, they are otherwise unimportant to the combinatorial structure we analyze and can be regarded as fixed and given (and thus ignored) for the remainder of the discussion.
Let
$$
   Q_t(Y^0, X,\theta) := |\Pi_t(Y^0, X,\theta)|
$$
denote the number of elements in the set $\Pi_t(Y^0, X, \theta)$, that is, the number of distinct values that the index $\pi_t((y_{t-1}, y_{t-2}, \ldots, y_1, Y^0), X_t, \theta)$ can take across different outcome histories. By construction, we have $Q_t(Y^0, X, \theta) \leq 2^{t-1}$, but in many practical models we have limited dependence on past outcomes and $Q_t(Y^0, X, \theta)$ is therefore often (much) smaller. Also, in most models, $Q_t(Y^0, X, \theta)$ remains constant across typical values of $Y^0$, $X$, and $\theta$, but may be different at specific points (e.g., when $\theta = 0$).

Next, define the set of all possible linear combinations of the $w_t$'s, where each coefficient $k_t$ is an integer between 0 and $Q_t(Y^0, X, \theta)$:
\begin{align*}
   {\cal D}(Y^0,X,\theta) := \left\{   \sum_{t=1}^T k_t \, w_t  \, : \,
   (k_1,\ldots,k_T) \in \prod_{t=1}^T \Big\{ 0, 1, \ldots, Q_t(Y^0, X, \theta) \Big\}
   \right\} \subset \mathbb{R}^{d_w}.
\end{align*}
The remark and examples below will help clarify the intuition behind the definition of the set $\mathcal{D}(Y^0, X, \theta)$. As before, we write $|\mathcal{D}(Y^0, X, \theta)|$ to denote its cardinality. The following result provides a lower bound on the number of linearly independent fixed-effect-free moment conditions that exist in the class of dynamic models considered in this section.

\begin{theorem}
   \label{th:ExistenceMoments}
   Let the model be given by
   Assumption~\ref{ass:DynamicModel},
   and let $Y^0$, $X$, $\theta$ be given values of the initial conditions, covariates, and common parameters.
   Then, the number of linearly independent
   moment functions of the form \eqref{MomentFunctions}
   is at least equal to
   $$
      2^T - \left|{\cal D}(Y^0,X,\theta) \right| .
   $$
   In particular, if 
   $\left|{\cal D}(Y^0,X,\theta) \right| < 2^T$,
   then non-zero moment functions of the form 
   \eqref{MomentFunctions} exist.
\end{theorem}

Theorem~\ref{th:ExistenceMoments} provides a sufficient condition for the existence of non-zero moment functions $m(Y, Y^0, X, \theta)$ satisfying \eqref{MomentFunctions}. However, the result does not characterize the structure of these functions, nor does it guarantee that they depend on $\theta$ in a way that allows identification. In practice, constructing such moment functions, and verifying that they identify $\theta$, is typically model-specific and requires additional algebraic or numerical insight. Concrete examples are discussed below.

Even so, the existence result in Theorem~\ref{th:ExistenceMoments} is useful. For example, once existence is guaranteed, one can apply the general functional-analytic framework developed by \citet{bonhomme2012functional} to compute the moment conditions numerically. In this ``functional differencing'' approach, moment functions can be approximated even when closed-form expressions are unavailable.

Whether obtained analytically or numerically, valid moment functions $m(Y, Y^0, X, \theta)$ can be used to construct GMM estimators that are root-$n$ consistent under standard regularity conditions.

\begin{remark}
The static logit model provides a useful special case for understanding the structure and role of the set $\mathcal{D}(Y^0, X, \theta)$. In that model, we have $Q_t(Y^0, X, \theta) = 1$ for all $t$, so the set reduces to

$$
\mathcal{D} = \left\{ \sum_{t=1}^T y_t w_t : y \in \{0,1\}^T \right\},
$$
which is simply the set of all possible values that the sufficient statistic $\mathsf{W} Y$ can take. As discussed in Section~\ref{sec:Static}, identification in the static model hinges on whether the mapping $Y \mapsto \mathsf{W} Y$ is injective. If $\mathsf{W} Y$ takes a distinct value for every realization $y \in \{0,1\}^T$, i.e., if $|\mathcal{D}| = 2^T$, then the existence of the sufficient statistic is not useful for identification or estimation of $\theta$, and no fixed-effect-free moment conditions exist. This explains the condition $|\mathcal{D}| < 2^T$ in Theorem~\ref{th:ExistenceMoments}, which guarantees the existence of nontrivial moment functions precisely when the sufficient statistic does not uniquely index every outcome configuration --- that is, when different realizations of $Y$ map to the same value of $\mathsf{W} Y$.

More generally, when sufficient statistics are not available, as in most dynamic models with general covariate values $X$, the set $\mathcal{D}(Y^0, X, \theta)$ still plays a closely analogous role. Although the linear combinations $\sum_{t=1}^T k_t w_t$ that define $\mathcal{D}$ can no longer be interpreted as sufficient statistics (since the coefficients $k_t$ are not restricted to binary values), they retain a similar structure and capture aspects of how the model links outcome histories to the fixed effects. The condition $|\mathcal{D}| < 2^T$ remains unchanged, but now serves only as a sufficient condition for the existence of nontrivial moment functions. In this sense, Theorem~\ref{th:ExistenceMoments} extends the logic of sufficiency-based identification to dynamic settings where no sufficient statistics exist.
\end{remark}

\subsection{Examples}

\subsubsection{Panel AR($p$) models with general covariates}

We first consider the class of AR(\(p\)) panel logit models with covariates and scalar fixed effects. These models are covered by Assumption~\ref{ass:DynamicModel}, with the index and fixed effect specification given by
\begin{align*}
\pi_t(Y^{t-1}, X_t, \theta) &= \sum_{r=1}^p  Y_{t-r} \, \gamma_r +  X_t' \, \beta,
&
w_t=1 ,
\end{align*}
where $\theta = (\gamma_1, \ldots, \gamma_p,\beta')'$. The initial conditions $Y^0 = (Y_{-p+1}, \ldots, Y_0) \in \{0,1\}^p$ are treated as fixed and observed.

To apply Theorem~\ref{th:ExistenceMoments}, we compute the number of distinct values that the index $\pi_t(Y^{t-1}, X_t, \theta)$ can take across binary outcome histories. Assuming $\gamma_r \neq 0$ for all $r \in \{1,\ldots,p\}$, this number is
$$
    Q_t = 2^{\min(p,\,t-1)},
$$
since the index depends on the most recent binary $p$ outcomes, or fewer in the initial periods.
Because we are considering $w_t = 1$ for all $t$, the set $\mathcal{D} = \mathcal{D}(Y^0, X, \theta)$ consists of all integers between 0 and the maximum sum $d_{\max} := \sum_{t=1}^T Q_t$. Therefore, we have
\begin{align*}
  |\mathcal{D} | &= 1 + \sum_{t=1}^T Q_t 
  \\
     &= 1 + \sum_{t=1}^T 2^{\min(p,\,t-1)}
     = \underbrace{1+\sum_{t=1}^{p} 2^{t-1}}_{=2^p} + 2^p \, (T - p)  
     \\
     &= 2^p (T + 1 - p).
\end{align*}
By Theorem~\ref{th:ExistenceMoments}, the number of linearly independent fixed-effect-free moment functions is therefore at least
$$
2^T - 2^p (T + 1 - p),
$$
which is strictly positive whenever $T \geq 2+p$ (in addition to $p$ observed time periods as initial conditions).
It turns out that this is not merely a lower bound: for generic values of the covariates $X$, this is in fact the exact number of linearly independent moment conditions in the AR(\(p\)) panel model. However, as we will see in the following example, additional moment conditions may become available for specific values of $X$, in particular when $X = 0$.

The result for the number of linearly independent moment conditions in the AR($p$) model was stated in \citet{honore2024dynamic}, who also derive explicit analytical expressions for the corresponding moment functions in the cases $p = 1, 2, 3$, and provide a detailed analysis of identification and estimation for $p = 1$, building on the results of \citet{kitazawa2021transformations}. The sharpness of the moment count has been established for $p = 1$ and $p = 2$ by \citet{kruiniger2020further}, and further discussed in \citet{dobronyi2021identification}. For general $p$, \citet{dano2023transition} proves that the bound is sharp under free-varying covariates and provides analytical expressions for all the moment functions in closed form.

\subsubsection{Panel AR(2) model without covariates}

We now consider a special case of the previous example where $p = 2$ and no covariates are present. The model takes the form
$$
\Pr(Y_t = 1 \mid Y_{t-1}, Y_{t-2}, A) = 
\frac{ \exp\left( Y_{t-1} \gamma_1 + Y_{t-2} \gamma_2 + A \right) }
{1 + \exp\left( Y_{t-1} \gamma_1 + Y_{t-2} \gamma_2 + A \right)},
$$
which corresponds to the structure in Assumption~\ref{ass:DynamicModel} with
\begin{align*}
\pi_t(Y^{t-1}, X_t, \theta) &= Y_{t-1} \gamma_1 + Y_{t-2} \gamma_2,
&
w_t &= 1.
\end{align*}
In this model for $T = 3$, one can construct one valid fixed-effect-free moment function for each value of the inital condition $Y^0 = (Y_{-1}, Y_0) \in \{0,1\}^2$. For instance, when $y^0 = (0,0)$, the following function satisfies the moment condition \eqref{MomentFunctions}:
$$
m(y, y^0, \theta) = 
\begin{cases}
\exp(-\gamma_1) & \text{if } y = (0,1,1), \\
1 & \text{if } y = (0,1,0), \\
-1 & \text{if } (y_1, y_2) = (1,0), \\
0 & \text{otherwise},
\end{cases}
$$
and when $y^0 = (0,1)$, a valid moment function is given by
$$
m(y, y^0, \theta) = 
\begin{cases}
\exp(\gamma_2 - \gamma_1) & \text{if } y = (1,0,0), \\
\exp(\gamma_2) & \text{if } y = (1,0,1), \\
-1 & \text{if } (y_1, y_2) = (0,1), \\
0 & \text{otherwise},
\end{cases}
$$
see  Section~3.3 of the first arXiv version of \citet{honore2024dynamic}, who provide closed-form moment functions for the AR($p$) model with $T = 3$ and $X_2=X_3$. 

Firstly, this example shows that the lower bound in Theorem~\ref{th:ExistenceMoments} is not always sharp. For $p = 2$ and $T = 3$, we have $2^T - |\mathcal{D}| = 0$, so the theorem does not guarantee the existence of any fixed-effect-free moment functions. Nevertheless, such functions do exist, as demonstrated above.

Secondly, the first moment function implies that $\gamma_1$ is point-identified in this model, since $\exp(-\gamma_1)$ is strictly monotonic in $\gamma_1$, and thus the moment function is strictly monotonic as well. Once $\gamma_1$ is identified, the second moment function then identifies $\gamma_2$ through the same logic.
This result provides analytical confirmation of earlier numerical findings in \citet{honore2019identification} that suggested identification might be possible in this setup, even though conditional likelihood methods do not apply.

Thus, the moment condition approach point-identifies both $\gamma_1$ and $\gamma_2$ in this model as soon as $T \geq 3$. By contrast, as noted in Remark~\ref{Remark:ARp:conditionalMLE}, the sufficient statistics approach can identify only $\gamma_2$, and requires at least $T \geq 4$ to do so.

\subsubsection{Panel AR(1) with quarterly fixed effects and covariates}
Next, consider the extension of Example \ref{ExampleQuarterlyFixedEffects} with strictly exogenous covariates. For this specification, we have $\pi_t(Y^{t-1}, X_t, \theta) =  Y_{t-1} \gamma + X_t' \beta$ and $w_t=(w_{t,1},w_{t,2},w_{t,3},w_{t,4})$ where $w_{t,q}$ is an indicator for quarter $q\in \{1,2,3,4\}$. It follows that 
\begin{align*}
\mathcal{D}(Y^0, X,\theta) &=  \Bigg\{
(k_1+\sum_{\substack{ t=2 \\ t\equiv 1 (\text{mod } 4)}}^T k_t, 
\sum_{\substack{ t=2 \\ t\equiv 2 (\text{mod } 4)}}^T k_t, 
\sum_{\substack{ t=2 \\ t\equiv 3 (\text{mod } 4)}}^T k_t, 
\sum_{\substack{ t=2 \\ t\equiv 0 (\text{mod } 4)}}^T k_t) : \\
&  k_1 \in \{0,1\} \text{ and } k_t \in \{0,1,2\} \; \forall t \geq 2
\Bigg\}
\end{align*}
 with cardinality $|\mathcal{D}(Y^0, X,\theta)| = \left(2 \lfloor \frac{T - 1}{4}\rfloor+2\right)\left(2 \lfloor \frac{T - 2}{4}\rfloor+3\right)\left(2 \lfloor \frac{T - 3}{4}\rfloor+3\right)\left(2 \lfloor \frac{T}{4}\rfloor+1\right)$. Hence, the bound  $2^T-|\mathcal{D}(Y^0, X,\theta)|$ of Theorem \ref{th:ExistenceMoments} becomes positive as soon as $T\geq 12$, ensuring the existence of moments. While informative, this bound is conservative in this instance since in the absence of regressors (which does not alter the bound here), Example \ref{ExampleQuarterlyFixedEffects} already indicated the existence of identifying moments for $T=6$. 
Indeed, one can construct two linearly independent moment functions $m_{1}$ and $m_{2}$ with $T=6$ periods. Using the shorthand $x_{ts}=x_{t}-x_{s}$ for $t\neq s$, $m_1$ is given by 
\begin{align*}
m_{1}(y,y_0,x,\theta)= \begin{cases}
  \exp(\gamma\left[(1-y_{0})+1\right]+(x_{51}+x_{26})'\beta) & \text{if } (y_1,y_2,y_4,y_5,y_6)=(1,0,1,0,1) \\
    \exp(\gamma(1-y_{0})+(x_{51}+x_{26})'\beta) & \text{if } (y_1,y_2,y_4,y_5,y_6)=(1,0,0,0,1) \\
    \exp(\gamma(1-y_{0})+x_{51}'\beta) & \text{if } (y_1,y_2,y_4,y_5,y_6)=(1,0,1,0,0) \\
    \exp(-\gamma y_{0}+x_{51}'\beta) & \text{if } (y_1,y_2,y_4,y_5,y_6)=(1,0,0,0,0) \\
    \exp(x_{26}'\beta)-1 & \text{if } (y_1,y_2,y_5,y_6)=(0,0,0,1) \\
    -1  & \text{if } (y_1,y_2,y_5)=(0,0,1) \\
    \exp(-\gamma y_{0}+(x_{51}+x_{62})'\beta) & \text{if } (y_1,y_2,y_4,y_5,y_6)=(1,1,1,0,0) \\
    \exp(-\gamma (1+y_{0})+(x_{51}+x_{62})'\beta) & \text{if } (y_1,y_2,y_4,y_5,y_6)=(1,1,0,0,0) \\
    \exp(\gamma (1-y_{0})+x_{51}'\beta) & \text{if } (y_1,y_2,y_4,y_5,y_6)=(1,1,1,0,1) \\
    \exp(-\gamma y_{0}+x_{51}'\beta) & \text{if } (y_1,y_2,y_4,y_5,y_6)=(1,1,0,0,1) \\
    \exp(x_{62}'\beta)-1 & \text{if } (y_1,y_2,y_5,y_6)=(0,1,0,0) \\
    -1 & \text{if } (y_1,y_2,y_5)=(0,1,1) \\
    0 & \text{otherwise, }
\end{cases}
\end{align*}and $m_{2}(y, y_0,x,\theta)$ is obtained by substituting $y_{t}$ by $(1-y_{t})$ for $t=0,\ldots, 6$ and $x$ by $-x$ in the expression of $m_{1}(y, y_0,x,\theta)$. In other words, $m_{2}(y, y_0,x,\theta)= m_{1}(1-y, 1-y_0,-x,\theta)$. We refer readers to Appendix \ref{moments_AR1_quarterlyFE} for detailed derivations of these expressions. Remark that the two moment functions depend on the common parameter $\theta=(\gamma,\beta')'$.

\subsubsection{Moment conditions in \citet{graham2016homophily} with additional covariates}
As an illustration of how the lower bound of Theorem~\ref{th:ExistenceMoments} can be applied to probe the existence of moment conditions in networks, consider the extension of \citet{graham2016homophily} introduced earlier, incorporating strictly exogenous covariates $X$. In this setting, we specify $\pi_t(Y^{t-1}, X_t, \theta) =  Y_{ij,\tau-1} \gamma + R_{ij,\tau-1}  \delta + X_t' \beta$, where $t=(i,j,\tau)$ indexes dyad-time pairs, and $w_t$ denotes the basis vector in $\mathbb{R}^{\binom{n}{2}}$ with entry one for dyad $(i,j)$, and zeros elsewhere. Under this formulation, the model described in Assumption~\ref{ass:DynamicModel} becomes
\begin{align*}
\Pr(Y_{ij\tau} = 1 \mid Y_{ij,\tau-1}, R_{ij,\tau-1},X_t ,A_{ij}) = 
\frac{ \exp\left(  Y_{ij,\tau-1}\gamma +  R_{ij,\tau-1}\delta + X_t' \beta+ A_{ij} \right) }
{ 1 + \exp\left(  Y_{ij,\tau-1}\gamma +  R_{ij,\tau-1}\delta + X_t' \beta+ A_{ij} \right) }
\end{align*}
Since $R_{ij,\tau-1} := \sum_{k \neq i,j} Y_{ik,\tau-1} Y_{jk,\tau-1}$, for any fixed covariates $X$,  the structure of $\pi_t$ implies that it can take at most $2(n-1)$ distinct values. Consequently, we have 
\begin{align*}
    Q_t(Y^0, X,\theta) = 
    \begin{cases}
        1 & \text{if } \tau = 1 \\
        2(n - 1) & \text{if } \tau \geq 2
    \end{cases}
\end{align*} and hence
\begin{align*}
   {\cal D}(Y_0, X) 
   = \Bigg\{ \sum_{t=1}^T k_t \, w_t, \text{ such that } 
   &\forall t = (i, j, \tau), \\
   &\text{if } \tau = 1, \quad k_t \in \{0,1\}, \\
   &\text{if } \tau \geq 2, \quad k_t \in \{0,1,\ldots, 2(n-1)\} 
   \Bigg\}.
\end{align*}
A straightforward counting argument then gives $|\mathcal{D}(Y^0, X,\theta)| = \left[2(\mathcal{T}-1)(n-1)+2\right]^{\binom{n}{2}}$ and Theorem~\ref{th:ExistenceMoments} implies the existence of at least $2^T - \left|{\cal D}(Y_0,X) \right|=2^T-\left[2(\mathcal{T}-1)(n-1)+2\right]^{\binom{n}{2}}$ moment functions free from the fixed effects. Since $\log |\mathcal{D}(Y^0, X,\theta)|$ grows approximately proportionally to  $\log \mathcal{T}$ - the number of time periods - while $\log(2^T)=\log(2)\binom{n}{2}\mathcal{T}$ grows linearly in $\mathcal{T}$, the existence of moment conditions is guaranteed for sufficiently large $\mathcal{T}$. For example, in the simplest case with $n=3$ agents, the lower bound is positive as soon as $T\geq 4$. In fact, \cite{dano2023transition} shows that there is generally as much as $\binom{n}{2}$ fixed-effect-free moment conditions with only $T=3$ periods that are explicitly given by:
\begin{align*}
    m_{y}(Y, Y^0,X,\theta) &= \mathds{1}\{Y_2=y\}\exp\left\{\sum_{i<j}(Y_{ij,3}-y_{ij})\left[\gamma (Y_{ij,1}-y_{ij}) + \delta (R_{ij,1}-r_{ij}) - (X_{ij,3}-X_{ij,2})' \beta\right]\right\} \\
    &\times \exp\left\{-\sum_{i<j}(Y_{ij,1}-y_{ij})\left[\gamma (Y_{ij,0}-y_{ij}) + \delta (R_{ij,0}-r_{ij}) - (X_{ij,3}-X_{ij,1})' \beta\right]\right\} \\
    &-\mathds{1}\{Y_1=y\}
\end{align*}
where $y$ denotes an undirected network, and $r_{ij}:= \sum_{k \neq i,j} y_{ik} y_{jk}$ denotes the number of friends that agents $i$ and $j$ share in common in $y$.

\subsection{Main Takeaways}

\begin{itemize}

\item
The moment function approach developed in this section --- also known as functional differencing, following \citet{bonhomme2012functional} --- provides a powerful alternative to the conditional likelihood strategy based on sufficient statistics. In particular, the examples above show that fixed-effect-free moment conditions exist even in dynamic models with arbitrary covariate variation, where sufficient statistics approaches typically fail. This includes models with general time-varying covariates, heterogeneous time trends, and rich dynamic structures.

\item
Theorem~\ref{th:ExistenceMoments} offers a general and easy-to-verify sufficient condition for the existence of such fixed-effect-free moment functions. While the lower bound it provides is not always sharp, it guarantees the existence of moment conditions for sufficiently large $T$ in a broad class of models. For example, in the AR(2) model without covariates, the theorem does not predict the existence of moments at $T=3$, but still confirms their existence for $T\geq 4$.

\item
Even when moment conditions are known to exist, their explicit construction and the demonstration that they identify the parameters of interest typically require model-specific derivations. In some models, analytical expressions can be derived, while in others, numerical methods, as in \citet{bonhomme2012functional}, may be required.

\item
Unlike the sufficiency-based approach, which relies on conditional likelihood and is typically estimated via conditional MLE, the moment-based approach leads naturally to GMM estimation. Once a sufficient number of valid moment conditions are constructed, and the model parameters are identified, standard GMM theory yields root-$n$ consistent and asymptotically normal estimators under appropriate regularity conditions.

\end{itemize}

\section{Conclusions}
\label{sec: Conclusion}

This paper has reviewed and extended two approaches for eliminating fixed effects in logit models: the conditional likelihood method and the construction of moment conditions. While the results are conceptually clean and methodologically promising, several important challenges remain, pointing to avenues for future research.

First, the moment-based framework often yields a large number of conditional moment conditions, linking it to the broader literature on optimal moment selection in panel data. Prior work, including \citet{bekker1994alternative}, \citet{donald2001choosing}, \citet{alvarez2003time}, and \citet{okui2009optimal}, has shown that using too many valid moments can introduce substantial finite-sample bias. This suggests that future research on nonlinear models should integrate identification strategies with principled approaches to moment selection.

Second, our focus has been on identifying and estimating structural parameters, rather than computing counterfactuals or marginal effects. In panel data settings, these quantities are typically not point-identified, even when the structural parameters are. Extending the methods discussed here to incorporate bounds on marginal effects, such as those proposed by \cite{PakelWeidner2024} and \cite{davezies2024identification}, would enhance the empirical relevance of the fixed effects framework.

Finally, and perhaps most critically, the models considered here assume strict exogeneity of the explanatory variables (aside from lagged outcomes). This assumption is often unrealistic in economic applications. While recent work by \cite{ArellanoCarrasco2003}, \cite{botosaru2024adversarial}, \cite{bonhomme2023identification}, and \cite{BonhommeDanoGraham2025} has made progress in relaxing this assumption, much remains to be done to develop robust methods that accommodate predetermined regressors in nonlinear panel models.

\appendix

\section{Proofs}

\begin{proof}[Proof of Lemma~\ref{lemma:Wperp}]
   Assume first that there exist $y_1$, $y_2$
   satisfying the assumptions in the lemma.
   In that case $w_{\perp}=y_1 - y_2$ satisfies 
   $w_{\perp} \in \{-1,0,1\}^T$ and $\mathsf{W} w_{\perp} = 0$.
   We have thus shown the 
   ``only if direction'' of the lemma.

   Conversely, let $w_{\perp} \in \{-1,0,1\}^T$ be such that $w_{\perp} \neq 0$ and $\mathsf{W} w_{\perp} = 0$.
   Define $y_1$ and $y_2$ to be the 
   $T$-vectors  with components $y_{1t}=\mathbbm{1}(w_{\perp,t}=1)$
   and $y_{2t}=\mathbbm{1}(w_{\perp,t}=-1)$,
   for $t \in \{1,\ldots,T\}$.
   Since $w_{\perp} \neq 0$ we have  $y_1 \neq y_2$.
   Since the definition of $y_1$ and $y_2$ implies $w_{\perp}=y_1-y_2$, and
    $\mathsf{W} w_{\perp} = 0$,
   we have $\mathsf{W}y_1 = \mathsf{W} y_2$.  We have thus shown the 
   ``if direction'' of the lemma.
\end{proof}

\begin{remark}
Notice that in the last paragraph of the proof, the 
   choice of $y_1$ and $y_2$ was somewhat
   arbitrary: For $t \in \{1,\ldots,T\}$
   with $w_{\perp,t}=0$ we chose 
   $y_{1t} = y_{2t} = 0$, but 
   $y_{1t} = y_{2t}$ (equal to zero or one) would have been  sufficient to deliver
   the desired result. 
   In other words, $w_{\perp,t}$ generally does not
   determine  $y_1$ and $y_2$ uniquely. However, purely from an identification perspective
   (completely ignoring finite sample properties), having multiple pairs $y_1$, $y_2$
   corresponding to the same $w_{\perp,t}$ is actually not useful, because all those pairs identify
   the same parameter component $\beta' X w_{\perp,t}$. In that sense, from an identification perspective, the $w_{\perp}$ can
   be viewed as more fundamental objects than the pairs $y_1$, $y_2$.
\end{remark}

\begin{proof}[Proof of Theorem~\ref{th:IdentificationStatic}]
   Let $w_\perp \in \{-1,0,1\}^T$ be a column of $\mathsf{W}_{\perp}$, and let $y_1, y_2 \in \{0,1\}^T$ be a corresponding outcome pair as guaranteed by Lemma~\ref{lemma:Wperp}. By construction, we have $w_\perp = y_1 - y_2$. Then, using equation~\eqref{SufficientStats1}, we obtain
$$
   \log\left(
   \frac{1}
   {\left[{\rm Pr} \left( Y = y_1 \, \big| \, X, Y \in \{y_1, y_2\} \right)\right]^{-1} - 1}
   \right)
   = \beta' X w_\perp,
$$
which shows that $\beta' X w_\perp$ is identified from the data.
Since this holds for every column of $\mathsf{W}_\perp$, it follows that the vector
$
z(X) := \beta' X \mathsf{W}_\perp
$
is point-identified from the distribution of the data.

Now suppose, contrary to identification, that there exists an alternative parameter vector $\widetilde \beta \neq \beta$ such that all assumptions of the theorem are satisfied. Then by the same logic, we must also have
$
z(X) = \widetilde \beta' X \mathsf{W}_\perp,
$
which implies
$$
(\beta - \widetilde \beta)' X \mathsf{W}_\perp = 0 \quad \text{a.s.}
$$
Let $b := \beta - \widetilde \beta$. Then the above equation reads $b' X \mathsf{W}_\perp = 0$ almost surely, contradicting the non-collinearity condition assumed in the theorem.
Therefore, no such $\widetilde \beta \neq \beta$ can exist, and $\beta$ is identified.
\end{proof}   

\begin{proof}[\bf Proof of Lemma~\ref{lemma:GeneralDynamicsSufficientStats}]
 Write $\pi_t$ for $\pi_t(y^{t-1},\theta)$, and $\tilde \pi_t$
for $\pi_t(\tilde y^{t-1},\theta)$.
By assumption (i) of the lemma, the likelihood ratio simplifies to:
$$
\frac{{\rm Pr}(Y=y | Y^0=y^0, A)}{{\rm Pr}(Y=\tilde y | Y^0=y^0, A)} =  \exp\left\{   \sum_{t=1}^{T}\left[ y_{t} \, \pi_t(y^{t-1},\theta) - \tilde y_{t} \, \pi_t(\tilde y^{t-1},\theta) \right] \right\} \prod_{t=2}^T \frac{1+\exp( \tilde \pi_t + w_t' A)}{1+\exp(\pi_t + w_t' A)} .
$$
Under assumption (ii) of the lemma, all the terms in the last product cancel pairwise due to the permutation condition, that is,
$$
 \prod_{t=2}^T \frac{1+\exp( \tilde \pi_t + w_t' A)}{1+\exp(\pi_t + w_t' A)} = 1 .
$$ 
Combining the last two displays gives the desired result.
\end{proof}

\begin{proof}[\bf Proof of Theorem~\ref{th:AR1Identification}]
The theorem is almost an immediate corollary of Lemma~\ref{lemma:GeneralDynamicsSufficientStats}. The only subtlety is that for part (i) of the theorem, we need to show that assumption (ii) of Lemma~\ref{lemma:GeneralDynamicsSufficientStats} holds for this model, that is, that 
$\Big[ \big(w_t,y_{t-1}\big) \, : \, t=2,\ldots,T\Big] $ 
is a permutation of  
$\Big[ \big(w_t,\widetilde y_{t-1}\big) \, : \, t=2,\ldots,T\Big]$.
To establish this, note that the sufficient statistic $(\mathsf{W}Y, \mathsf{W}Y_{\rm lag})$ provides us with the constraints:
\begin{align}
\sum_{t=1}^T w_t y_t &= \sum_{t=1}^T w_t \tilde{y}_t, \label{constraint1}\\
\sum_{t=1}^T w_t y_{t-1} &= \sum_{t=1}^T w_t \tilde{y}_{t-1}. \label{constraint2}
\end{align}
The key insight is that each $w_t$ is a standard basis vector in $\mathbb{R}^{d_w}$ due to the restrictions in \eqref{RestrictionWbinary}. Specifically, we can write $w_t = e_{j_t}$ where $j_t \in \{1,\ldots,d_w\}$ and $e_j$ is the $j$-th standard basis vector.
Under this representation, constraints \eqref{constraint1} and \eqref{constraint2} become:
\begin{align}
\text{For each } j \in \{1,\ldots,d_w\}: \quad \sum_{t: j_t = j} y_t &= \sum_{t: j_t = j} \tilde{y}_t, \label{group_constraint1}\\
\text{For each } j \in \{1,\ldots,d_w\}: \quad \sum_{t: j_t = j} y_{t-1} &= \sum_{t: j_t = j} \tilde{y}_{t-1}. \label{group_constraint2}
\end{align}
Since $y_t, \tilde{y}_t \in \{0,1\}$, constraints \eqref{group_constraint1} and \eqref{group_constraint2} imply that for each group of time periods with the same $w_t$ value (i.e., for each $j \in \{1,\ldots,d_w\}$):
\begin{itemize}
\item The number of ones in $\{y_t : t \in \{1,\ldots,T\}, j_t = j\}$ equals the number of ones in $\{\tilde{y}_t : t \in \{1,\ldots,T\}, j_t = j\}$.
\item The number of ones in $\{y_{t-1} : t \in \{2,\ldots,T\}, j_t = j\}$ equals the number of ones in $\{\tilde{y}_{t-1} : t \in \{2,\ldots,T\}, j_t = j\}$.
\end{itemize}

This means that within each group defined by $w_t = e_j$, the binary values can be rearranged such that the pairs $(w_t, y_{t-1})$ from the $y$ sequence match the pairs $(w_t, \tilde{y}_{t-1})$ from the $\tilde{y}$ sequence in terms of their frequency distribution.

Since this matching property holds for each group separately, and the groups partition the index set $\{2,\ldots,T\}$, we conclude that the sequences $\Big[ \big(w_t,y_{t-1}\big) \, : \, t=2,\ldots,T\Big]$ and $\Big[ \big(w_t,\tilde{y}_{t-1}\big) \, : \, t=2,\ldots,T\Big]$ are permutations of each other.

Therefore, condition (ii) from Lemma~\ref{lemma:GeneralDynamicsSufficientStats} is satisfied. Combined with the fact that condition (i) of the lemma follows directly from constraint \eqref{constraint1}, we can apply the lemma to establish that the distribution ratio does not depend on $A$, which proves part (i) of the theorem.

Part (ii) follows immediately from the identification condition in the lemma and the specific form $\pi_t(y^{t-1},\theta) = y_{t-1}\gamma$ in the AR(1) model.
\end{proof}

\begin{proof}[Proof of Theorem~\ref{th:ExistenceMoments}]
We drop most arguments $Y_0$, $X$ throughout this proof.
We also use 
$y^{t-1}$ simply to denote the vector
$(y_{t-1},y_{t-2},\ldots,y_1)$
(i.e. dropping $Y^0$).
Define $a_t:=\exp(w_t' \, A)$
and $b_{t,q}:=\exp(\pi_{t,q})$,
where $\pi_{t,q}$ denotes the elements of $\Pi_t$, with $q \in \{1,\ldots,Q_t\}$.
Also, let $q(y^{t-1})$ be such that 
$b_{t,q(y^{t-1})}=\exp(\pi_t(y^{t-1}))$.
Then, 
\begin{align*}
  {\rm Pr} \left( Y=y \, \big| \, Y^0,X,A \right)  &= \prod_{t=1}^T \frac {\left( b_{t,q(y^{t-1})} \, a_t \right)^{y_t}} {1+b_{t,q(y^{t-1})} \, a_t  }
  \\
  &= 
  \underbrace{
  \left[\prod_{t=1}^T
  \prod_{q=1}^{Q_t}
     \frac 1 {{1+b_{t,q} \, a_t  }}
  \right]}_{=:\kappa(a)}
  \prod_{t=1}^T 
  \underbrace{
  (b_{t,q(y^{t-1})})^{y_t}  (a_t)^{y_t}
   \prod_{q\in\{1,\ldots,Q_t\} \setminus q(y^{t-1})}
     ({1+b_{t,q} \, a_t  })
     }_{=:\phi_t(y,a_t)}
  \\
  &= \kappa(a) \underbrace{
  \prod_{t=1}^T  \phi_t(y,a_t) }_{=:\phi(y,A)},
\end{align*}
where for each $t$ the corresponding 
\begin{align*}
   \phi_t(y,a_t) = \sum_{k=0}^{Q_t} c_{t,k}(y) a_t^{k}
\end{align*}
is a polynomial in $a_t$ with $Q_t+1$ powers between $a_t^0$ and $a_t^{Q_t}$. 
Using this we find that
\begin{align*}
    \phi(y,A) = \sum_{d \in {\cal D}}
    \tilde c_d(y)  \exp(d'A)
\end{align*}
Therefore, a moment function satisfies
\eqref{MomentFunctions} 
for all $A$ if it is orthogonal (in the $2^T$
dimensional outcome space)
to all $\tilde c_d$ vectors,
which is one linear condition on the moment
function for every $d \in {\cal D}$.
The number of solutions is therefore at least
$2^T - |{\cal D}|$.
\end{proof}

\subsection{Sufficient statistics for AR($p$) models with $p>1$} \label{Details_ARp}

Consider
\begin{align}       
&{\rm Pr} \left( Y=y \, \big| \, Y^0=y^0,A \right) \notag
\\
&=  \prod_{t=1}^T \, 
\left[ \frac {1} {1+\exp\left(\sum_{r=1}^p y_{t-r} \, \gamma_r + w_t' \, A\right) }  
\right]^{1-y_t}
\left[ \frac {\exp\left(\sum_{r=1}^p y_{t-r} \, \gamma_r + w_t' \, A\right)} {1+\exp\left(\sum_{r=1}^p y_{t-r} \, \gamma_r + w_t' \, A\right) }  
\right]^{y_t} . \label{ModelARp}
\end{align}
which extends \eqref{ModelAR2} to an autoregressive model of arbitrary order $p>1$. Let $y, \tilde{y}\in \{0,1\}^{T}$ denote two outcome paths. In this model, conditions (i) and (ii) of Lemma \ref{lemma:GeneralDynamicsSufficientStats} correspond to 
    \begin{itemize}
       \item[(i)]
    $\sum_{t=1}^T w_t y_t = \sum_{t=1}^T w_t \tilde{y}_t$.

    \item[(ii)]  
$\Big[ \big(w_t,\sum_{r=1}^p y_{t-r}\, \gamma_r\big) \, : \, t=2,\ldots,T\Big] $ 
is a permutation of  
$\Big[ \big(w_t,\sum_{r=1}^p y_{t-r}\, \gamma_r\big) \, : \, t=2,\ldots,T\Big]$
    \end{itemize}
Let $\mathcal{C}_{l,p}$ denote the set of all combinations of $l$-elements $\bold{i}=(i_1,\ldots,i_l)$ drawn from $\{1,\ldots,p\}$. A key observation is that condition (ii) means that  $y,\widetilde y^T$ produce the same proportions of pairs $\left\{\left(w_t,\sum_{j=1}^l \gamma_{i_j}\right)\right\}_{\bold{i}\in \mathcal{C}_{l,p}}$ for $l=1,\ldots,p$, $t=2,\ldots,T$. In view of \eqref{ModelARp}, this can only be the case if for each $\bold{i}\in \mathcal{C}_{l,p}$ and associated pair $\left(w_t,\sum_{j=1}^l \gamma_{i_j}\right)$, we have
\begin{align*}
    \sum_{t=2}^{T} w_t 
 \prod_{j=1}^l y_{t - i_j} 
 \prod_{\substack{k = 1 \\ k \notin \{i_1, \ldots, i_l\}}}^p (1 - y_{t - k}) &= \sum_{t=2}^{T} w_t 
 \prod_{j=1}^l \widetilde y_{t - i_j} 
 \prod_{\substack{k = 1 \\ k \notin \{i_1, \ldots, i_l\}}}^p (1 - \widetilde y_{t - k})
\end{align*}
These conditions are collectively equivalent to  
\begin{align} \label{SufficientStats_pureARp_part2}
    \sum_{t=2}^{T} w_t 
 \prod_{j=1}^l y_{t - i_j} 
  &= \sum_{t=2}^{T} w_t 
 \prod_{j=1}^l \widetilde y_{t - i_j}, \quad \forall \bold{i}\in \mathcal{C}_{l,p}, \quad l=1,\ldots,p
\end{align}
Lemma \ref{lemma:GeneralDynamicsSufficientStats} ensures that if (i) and \eqref{SufficientStats_pureARp_part2} are satisfied, the likelihood ratio 
$\frac{{\rm Pr}(Y=y | Y^0=y^0, A)}{{\rm Pr}(Y=\tilde y | Y^0=y^0, A)}$ is free from $A$.
For the special case $p=2$, we have $C_{1,2}=\{(1),(2)\}$ and $C_{2,2}=\{(1,2)\}$ and \eqref{SufficientStats_pureARp_part2} reduces to 
\begin{align} 
   \sum_{t=1}^T w_t \, y_{t-1} &= \sum_{t=1}^T w_t \, \tilde{y}_{t-1} \notag
   \\
   \sum_{t=1}^T w_t \, y_{t-2} &= \sum_{t=1}^T w_t \, \tilde{y}_{t-2} \notag
   \\
   \sum_{t=1}^T w_t \, y_{t-1} \, y_{t-2} &= \sum_{t=1}^T w_t \, \tilde y_{t-1} \, \tilde{y}_{t-2}. \notag
\end{align}
Together with (i), these restrictions coincide with \eqref{SufficientStats_pureAR2} from the main text. \\

To understand why the conditional likelihood approach fails to identify $\gamma_1$ in the AR(2),  recall that each weight vector $w_t = (w_{t,1}, \ldots, w_{t,d_w})$ satisfies
$w_{t,k} \in \{0,1\}$, for all $k \in \{1,\ldots,d_w\}$, and  $\sum_{k=1}^{d_w} w_{t,k} =1$. This structure implies that the conditions in \eqref{SufficientStats_pureAR2} yield the following equalities: (a) $\sum_{t=1}^T \, y_{t} = \sum_{t=1}^T \, \tilde{y}_{t}$ (b) $\sum_{t=1}^T \, y_{t-1} = \sum_{t=1}^T \, \tilde{y}_{t-1}$, (c) $\sum_{t=1}^T \, y_{t-2} = \sum_{t=1}^T \, \tilde{y}_{t-2}$ and (d) $\sum_{t=1}^T \, y_{t-1} \, y_{t-2} = \sum_{t=1}^T \, \tilde y_{t-1} \, \tilde{y}_{t-2}$. Furthermore, since the initial condition $y^0=(y_{0},y_{-1})\in\{0,1\}^2$ is held fixed across $y, \tilde{y}$, we  have (b) $\sum_{t=1}^{T-1} \, y_{t} = \sum_{t=1}^{T-1} \, \tilde{y}_{t}$, (c) $\sum_{t=1}^{T-2} \, y_{t} = \sum_{t=1}^{T-2} \, \tilde{y}_{t}$ and (d) $\sum_{t=1}^{T-1} \, y_{t} \, y_{t-1} = \sum_{t=1}^{T-1} \, \tilde y_{t} \, \tilde{y}_{t-1}$. From (b) and (c), we deduce that $y_{T-1}=\widetilde y_{T-1}$; then from (a) it follows that $y_{T}=\widetilde y_{T}$. Using these equalities in (d), we find: $\sum_{t=1}^{T} \, y_{t} \, y_{t-1} = \sum_{t=1}^{T} \, \tilde y_{t} \, \tilde{y}_{t-1}$. This last condition implies that the likelihood ratio given in Lemma \ref{lemma:GeneralDynamicsSufficientStats} simplifies to
\begin{align*}
    \frac{{\rm Pr}(Y=y | Y^0=y^0, A)}{{\rm Pr}(Y=\tilde y | Y^0=y^0, A)} = \exp\left\{  \gamma_2 \sum_{t=1}^{T}\left[ y_{t}y_{t-2} - \tilde y_{t}\tilde y_{t-2}\right] \right\},
\end{align*}
which is independent of $\gamma_1$. For the general AR($p$) case, one can extend this reasoning to show that the conditional likelihood approach fails to identify the first $(p-1)$ autoregressive coefficients $\gamma_1,\ldots,\gamma_{p-1}$ under conditions (i) and \eqref{SufficientStats_pureARp_part2}.

\subsection{Moment functions for the panel AR(1) with quarterly fixed effects and covariates} \label{moments_AR1_quarterlyFE}

The moment function $m_1$ presented in the main text in a very explicit form is the sum of 9 subcomponents:
\begin{align*}
   m_{1}(Y,Y_0,X,\theta) =  \sum_{y\in \{0,1\}^2} \varphi^{(a)}_{y}(Y,Y_0,X,\theta)+\varphi^{(b)}_{y}(Y,Y_0,X,\theta)-(1-Y_{1})
\end{align*}
where for $y=(y_3,y_4)\in \{0,1\}^2$,
\begin{align*}
    &\varphi^{(a)}_{y}(Y,Y_0,X,\theta)=\phi_{y}^{(a)}(Y,Y_0,X,\theta)-w_{y}(Y_0,X)Y_1\phi_{y}^{(a)}(Y,Y_0,X,\theta) \\
    &\phi_{y}^{(a)}(Y,Y_0,X,\theta) = (1-Y_2)\mathds{1}\{Y_3=y_3,Y_4=y_4\}(1-Y_5)e^{Y_6(\gamma Y_1+X_{26}'\beta)}  \\
     &\varphi^{(b)}_{y}(Y,Y_0,X,\theta)=\phi_{y}^{(b)}(Y,Y_0,X,\theta)-w_{y}(Y_0,X)Y_1\phi_{y}^{(b)}(Y,Y_0,X,\theta) \\
    &\phi_{y}^{(b)}(Y,Y_0,X,\theta) = Y_2\mathds{1}\{Y_3=y_3,Y_4=y_4\}(1-Y_5)e^{(1-Y_6)(-\gamma Y_1+X_{62}'\beta)}
\end{align*}
with $w_{y}(Y_0,X) = 1-e^{-\gamma Y_0 +\gamma y_{4}+X_{51}'\beta}$ and using the shorthand $X_{ts}=X_{t}-X_{s}$. The blueprint behind this construction is the following. For any $y=(y_3,y_4)$, it follows from the definition of $\phi_{y}^{(a)}(Y,Y_0,X,\theta)$ that
\begin{align*}
    \mathbb{E}\left[\phi^{(a)}_{y}(Y,Y_0,X,\theta)|Y_0,Y_1,X,A\right]&=\frac{1}{1+e^{\gamma Y_{1}+ X_{2}'\beta+A_2}}\frac{e^{ y_{3}(X_{3}'\beta+A_3)}}{1+e^{ X_{3}'\beta+A_3}}\frac{e^{y_4(\gamma y_{3}+X_{4}'\beta+A_4)}}{1+e^{\gamma y_{3}+X_{4}'\beta+A_4}}\frac{1}{1+e^{\gamma y_4+X_{5}'\beta+A_1}}\\
    &\left(\frac{1}{1+e^{X_{6}'\beta+A_2}}+e^{\gamma Y_1+X_{26}'\beta}\frac{e^{X_{6}'\beta+A_2}}{1+e^{X_{6}'\beta+A_2}}\right) \\
    &=\frac{e^{ y_{3}(X_{3}'\beta+A_3)}}{1+e^{ X_{3}'\beta+A_3}}\frac{e^{y_4(\gamma y_{3}+X_{4}'\beta+A_4)}}{1+e^{\gamma y_{3}+X_{4}'\beta+A_4}}\frac{1}{1+e^{\gamma y_4+X_{5}'\beta+A_1}}\frac{1}{1+e^{X_{6}'\beta+A_2}}
\end{align*}
Therefore, by the law of iterated expectations
\begin{align*}
    &\mathbb{E}\left[\varphi^{(a)}_{y}(Y,Y_0,X,\theta)|Y_0,X,A\right]\\
    &= \mathbb{E}\left[\phi_{y}^{(a)}(Y,Y_0,X,\theta)-w_{y}(Y_0,X)Y_1\phi_{y}^{(a)}(Y,Y_0,X,\theta)|Y_0,X,A\right] \\
    &=\mathbb{E}\left[\mathbb{E}\left[\phi_{y}^{(a)}(Y,Y_0,X,\theta)|Y_0,Y_1,X,A\right]|Y_0,X,A\right] \\
    &-w_{y}(Y_0,X)\mathbb{E}\left[Y_1\mathbb{E}\left[\phi_{y}^{(a)}(Y,Y_0,X,\theta)|Y_0,Y_1,X,A\right]|Y_0,X,A\right] \\
    &=\frac{e^{ y_{3}(X_{3}'\beta+A_3)}}{1+e^{ X_{3}'\beta+A_3}}\frac{e^{y_4(\gamma y_{3}+X_{4}'\beta+A_4)}}{1+e^{\gamma y_{3}+X_{4}'\beta+A_4}}\frac{1}{1+e^{\gamma y_4+X_{5}'\beta+A_1}}\frac{1}{1+e^{X_{6}'\beta+A_2}} \\
    &-\left(1-e^{-\gamma Y_0 +\gamma y_{4}+X_{51}'\beta}\right)\frac{e^{\gamma Y_0+X_{1}'\beta+A_1}}{1+e^{\gamma Y_0+X_{1}'\beta+A_1}}\frac{e^{ y_{3}(X_{3}'\beta+A_3)}}{1+e^{ X_{3}'\beta+A_3}}\frac{e^{y_4(\gamma y_{3}+X_{4}'\beta+A_4)}}{1+e^{\gamma y_{3}+X_{4}'\beta+A_4}}\frac{1}{1+e^{\gamma y_4+X_{5}'\beta+A_1}}\frac{1}{1+e^{X_{6}'\beta+A_2}} \\
    &=\left(\frac{1}{1+e^{\gamma y_4+X_{5}'\beta+A_1}}-\left(1-e^{-\gamma Y_0 +\gamma y_{4}+X_{51}'\beta}\right)\frac{e^{\gamma Y_0+X_{1}'\beta+A_1}}{1+e^{\gamma Y_0+X_{1}'\beta+A_1}}\frac{1}{1+e^{\gamma y_4+X_{5}'\beta+A_1}}\right) \\
    &\times \frac{e^{ y_{3}(X_{3}'\beta+A_3)}}{1+e^{ X_{3}'\beta+A_3}}\frac{e^{y_4(\gamma y_{3}+X_{4}'\beta+A_4)}}{1+e^{\gamma y_{3}+X_{4}'\beta+A_4}}\frac{1}{1+e^{X_{6}'\beta+A_2}} \\
    &=\frac{1}{1+e^{\gamma Y_0+X_{1}'\beta+A_1}}\frac{e^{ y_{3}(X_{3}'\beta+A_3)}}{1+e^{ X_{3}'\beta+A_3}}\frac{e^{y_4(\gamma y_{3}+X_{4}'\beta+A_4)}}{1+e^{\gamma y_{3}+X_{4}'\beta+A_4}}\frac{1}{1+e^{X_{6}'\beta+A_2}}
\end{align*}
where the last equality exploits the partial fraction decomposition identity (c.f Lemma 6 in \cite{dano2023transition})
\begin{align*} 
    \frac{1}{1+e^{v+a}}+(1-e^{u-v})\frac{e^{v+a}}{(1+e^{v+a})(1+e^{u+a})}=\frac{1}{1+e^{u+a}}
\end{align*}
Now, by summing over all possible combinations of $y=(y_3,y_4)$, we obtain:
\begin{align} \label{moment_quaterlyFE_parta}
    \mathbb{E}\left[\sum_{y\in \{0,1\}^2} \varphi^{(a)}_{y}(Y,Y_0,X,\theta)|Y_0,X,A\right]&=     \frac{1}{1+e^{\gamma Y_0+X_{1}'\beta+A_1}}\frac{1}{1+e^{X_{6}'\beta+A_2}}
\end{align}
Next, applying the same reasoning to $\phi^{(b)}_{y}(Y,Y_0,X,\theta)$ and $\varphi^{(b)}_{y}(Y,Y_0,X,\theta)$, we get
\begin{align*}
    \mathbb{E}\left[\phi^{(b)}_{y}(Y,Y_0,X,\theta)|Y_0,Y_1,X,A\right]
    &=\frac{e^{ y_{3}(\gamma+X_{3}'\beta+A_3)}}{1+e^{ \gamma+X_{3}'\beta+A_3}}\frac{e^{y_4(\gamma y_{3}+X_{4}'\beta+A_4)}}{1+e^{\gamma y_{3}+X_{4}'\beta+A_4}}\frac{1}{1+e^{\gamma y_4+X_{5}'\beta+A_1}}\frac{e^{X_{6}'\beta+A_2}}{1+e^{X_{6}'\beta+A_2}}
\end{align*}
and 
\begin{align*}
    \mathbb{E}\left[\varphi^{(b)}_{y}(Y,Y_0,X,\theta)|Y_0,X,A\right]
    &=\frac{1}{1+e^{\gamma Y_0+X_{1}'\beta+A_1}}\frac{e^{ y_{3}(\gamma+X_{3}'\beta+A_3)}}{1+e^{\gamma+X_{3}'\beta+A_3}}\frac{e^{y_4(\gamma y_{3}+X_{4}'\beta+A_4)}}{1+e^{\gamma y_{3}+X_{4}'\beta+A_4}}\frac{e^{X_{6}'\beta+A_2}}{1+e^{X_{6}'\beta+A_2}}
\end{align*}
implying in turn that
\begin{align} \label{moment_quaterlyFE_partb}
   \mathbb{E}\left[\sum_{y\in \{0,1\}^2} \varphi^{(b)}_{y}(Y,Y_0,X,\theta)|Y_0,X,A\right]&=     \frac{1}{1+e^{\gamma Y_0+X_{1}'\beta+A_1}}\frac{e^{X_{6}'\beta+A_2}}{1+e^{X_{6}'\beta+A_2}}
\end{align}
Adding up \eqref{moment_quaterlyFE_parta} and \eqref{moment_quaterlyFE_partb} finally yields
\begin{align*}
    \mathbb{E}\left[\sum_{y\in \{0,1\}^2} \varphi^{(a)}_{y}(Y,Y_0,X,\theta)+\varphi^{(b)}_{y}(Y,Y_0,X,\theta)|Y_0,X,A\right]&=     \frac{1}{1+e^{\gamma Y_0+X_{1}'\beta+A_1}} = \mathbb{E}\left[(1-Y_{1})|Y_0,X,A\right]
\end{align*}
whereupon $\mathbb{E}\left[m_{1}(Y,Y_0,X,\theta)|Y_0,X,A\right]=0$. To see that $m_{2}(Y,Y_0,X,\theta)=m_{1}(1-Y,-X;\theta)$ is also a valid moment function, it suffices to note that the model probabilities ${\rm Pr} \left( Y=y \, \big| \, Y^0,X,A \right)$ of the AR(1) model with quaterly fixed effects are invariant under the symmetry transformation:
\begin{align*}
    Y_t \leftrightarrow 1-Y_{t}, \quad X_{t}\leftrightarrow -X_{t}, \quad \theta \leftrightarrow \theta, \quad A_{k} \leftrightarrow -A_{k}-\gamma \quad \forall k \in \{1,\ldots,4\}
\end{align*}
This implies that
\begin{align*}
    \mathbb{E}\left[m_{2}(Y,Y_0,X,\theta)|Y_0=y_0,X=x,A=a\right]&=\mathbb{E}\left[m_{1}(1-Y,-x;\theta)|Y_0=y_0,X=x,A=a\right] \\
    &=\mathbb{E}\left[m_{1}(Y,Y_0,X,\theta)|Y_0=1-y_0,X=-x,A=-a-\gamma\right] \\
    &=0
\end{align*}

\end{document}